\documentclass[journal,twoside,web]{ieeecolor}
\pdfoutput=1
\usepackage{generic}
\usepackage{cite}
\usepackage{amsmath,amssymb,amsfonts}
\usepackage{mathtools}
\usepackage{graphicx}
\usepackage{algorithm,algorithmic}
\usepackage{hyperref}
\hypersetup{hidelinks=true}
\usepackage{textcomp}

\let\labelindent\relax
\usepackage{amsthm}
\usepackage{enumitem}
\usepackage{comment}
\usepackage[font=footnotesize]{subcaption}
\usepackage[font=footnotesize]{caption}

\newcommand{\xx}{\boldsymbol{x}}
\newcommand{\XX}{\mathcal{X}}
\newcommand{\balpha}{\boldsymbol{\alpha}}
\newcommand{\bell}{\boldsymbol{\ell}}
\newcommand{\bbeta}{\boldsymbol{\beta}}
\newcommand{\bg}{\boldsymbol{g}}

\newcommand{\bellstar}{\boldsymbol{\ell^*}}

\newcommand{\bzero}{\boldsymbol{0}}
\newcommand{\yy}{\boldsymbol{y}}
\newcommand{\xstar}{\boldsymbol{x^*}}
\newcommand{\xnstar}{\boldsymbol{x}_{n}^{\boldsymbol{*}}}
\newcommand{\xminusnstar}{\boldsymbol{x}_{-n}^{\boldsymbol{*}}}
\newcommand{\balphastar}{\boldsymbol{\alpha^*}}
\newcommand{\RR}{\mathbb{R}}
\newcommand{\MM}{\boldsymbol{M}}

\newcommand{\EE}{\mathbb{E}}
\newcommand{\FF}{\mathcal{F}}
\newcommand{\bb}{\mathfrak{b}}
\newcommand{\Nopt}{\mathcal{N}_{opt}}
\newcommand{\bA}{\boldsymbol{A}}

\newcommand{\ATalpha}{\bA^{\intercal}\balpha}
\newcommand{\AT}{\bA^{\intercal}}

\theoremstyle{definition}
\newtheorem{theorem}{Theorem}
\newtheorem{lemma}[theorem]{Lemma}
\newtheorem{proposition}[theorem]{Proposition}
\newtheorem{remark}{Remark}
\newtheorem{definition}{Definition}

\newtheorem{assumption}{Assumption}

\def\BibTeX{{\rm B\kern-.05em{\sc i\kern-.025em b}\kern-.08em
    T\kern-.1667em\lower.7ex\hbox{E}\kern-.125emX}}

\begin{document}
\title{Learning to Control Unknown Strongly Monotone Games}
\author{Siddharth Chandak, Ilai Bistritz, Nicholas Bambos
\thanks{Siddharth Chandak and Nicholas Bambos are with the Department of Electrical Engineering, Stanford University, Stanford, CA, USA. Ilai Bistritz is with the Department of Industrial Engineering and the School of Electrical Engineering, Tel Aviv Univeristy, Israel. \\
Emails: {\tt chandaks@stanford.edu,\\
ilaibistritz@tauex.tau.ac.il,\\
bambos@stanford.edu}}
}

\maketitle

\begin{abstract}
Consider a strongly monotone game where the players' utility functions include a reward function and a linear term for each dimension, with coefficients that are controlled by the manager. Gradient play converges to a unique Nash equilibrium (NE) that does not optimize the global objective. The global performance at NE can be improved by imposing linear constraints on the NE, also known as a generalized Nash equilibrium (GNE). We therefore want the manager to control the coefficients such that they impose the desired constraint on the NE. However, this requires knowing the players' rewards and action sets. Obtaining this game information is infeasible in a large-scale network and violates user privacy. To overcome this, we propose a simple algorithm that learns to shift the NE to meet the linear constraints by adjusting the controlled coefficients online. Our algorithm only requires the linear constraints violation as feedback and does not need to know the reward functions or the action sets. We prove that our algorithm converges with probability 1 to the set of GNE given by coupled linear constraints. We then prove an L2 convergence rate of near-$O(t^{-1/4})$. 
\end{abstract}

\section{Introduction}
Large-scale networks such as in communication, transportation, energy, and computing consist of a myriad of agents taking local decisions. Agents do not see the bigger picture and instead only optimize their local objectives. The result is often an inefficient equilibrium \cite{christodoulou2005price}. Hypothetically, a centralized authority that knows all network parameters can instruct the agents how to act to optimize a global objective. 
However, such a centralized architecture is infeasible for large-scale
networks and also violates the users' privacy. 

The interaction between the users is formalized as a game \cite{alpcan2002cdma,marden2013distributed}. The players may be cooperative, or they may be selfishly interested
in maximizing their reward. In large-scale systems, obtaining performance
guarantees is challenging even with cooperative players since they
do not know the reward functions of their peers.
A Nash Equilibrium (NE) aims to predict the outcome of such a repeated
interaction. In monotone games \cite{rosen1965existence,tatarenko2019learning,mertikopoulos2019learning}, which are our focus here, simple distributed algorithms such as gradient ascent of players on their reward functions are guaranteed to converge to a NE.

A NE does not optimize global objectives and can lead to poor performance. By controlling parameters in the utility functions of the players, the manager can aim to change the game such that the NE is efficient. In general, setting utility parameters to optimize performance at the NE is intractable. Therefore, in this paper, we focus on the task of imposing linear constraints on the resulting NE of the game the manager controls.  Linear constraints are common in many optimization formulations. When optimizing a global quadratic cost, linear constraints can describe the first-order optimality condition for the NE. In resource allocation, linear constraints can describe what should be the ideal loads at each resource at NE. 

Even when setting utility parameters to impose linear constraints at the NE is tractable, it requires
knowing the reward functions of players and their action sets. This
would require the manager to collect this information from the players,
and then solve a large-scale optimization problem. This centralized
approach violates the privacy of the users and is infeasible in large-scale
networks. Even worse, the manager would need to collect the parameters
periodically if they are time-varying, and then solve the large-scale
optimization problem again.

To converge to a NE that meets target linear constraints in an unknown strongly monotone
game, we propose an online learning approach. In our algorithm, the
manager adjusts the controlled coefficients online and receives the constraint violation as feedback. Such feedback is easy to monitor
in practice and preserves user privacy. Our algorithm is based on two time-scale stochastic approximation, where the agents update their actions in the faster time-scale and the manager updates the controlled coefficients in the slower time-scale. 

We prove that our algorithm
converges to the set of NE that satisfy the target linear constraints with probability 1. We further show that our algorithm has a mean square error rate of near-$O(t^{-1/4})$ for the linear constraints. Due to the convex and compact action sets, our slower time-scale is a fixed point iteration with a non-expansive mapping instead of the stronger contractive mapping that allows for a $O(t^{-1})$ convergence rate \cite{Chandak-TTS-Opti, chen2020finite}. Apart from stochastic approximation, we employ tools from variational inequalities and the Krasnosel’ski\u{i}–Mann (KM) algorithm in our analysis. Therefore, we provide a theoretically principled game control scheme for strongly monotone games.

\subsection{Related Work}

Distributed optimization \cite{nedic2009distributed,molzahn2017survey}
is concerned with optimizing a target objective over
a network of agents. However, these agents are non-interacting computational
units that are not coupled in a game structure. To cast the game
as a distributed optimization with the action profile as the variable,
each agent would need to know the effect of other agents on its reward, which is infeasible. Moreover, distributed optimization
requires significant communication between agents. 

Distributed algorithms that converge to NE have been studied in \cite{mazumdar2020gradient,frihauf2011nash,ye2017distributed}, and in  \cite{tatarenko2020geometric,mertikopoulos2019learning,gadjov2018passivity} for monotone games. The players in our game use simple stochastic gradient play to converge to NE. However, if the game is uncontrolled, the NE does not optimize global objectives. In our game control scheme, the manager learns to steer the game to an efficient NE, which is a GNE with target linear constraints.

Convergence to GNE \cite{GNEP-Facc-survey} in strongly monotone games is usually studied in the noiseless setting \cite{GNEP-Kim, GNEP-Jordan, GNEP-Aussel}, while our learning agents make the convergence to a GNE stochastic. Stochastic convergence to GNEs has recently started to attract attention. In \cite{GNEP-Franci-learning} players share their actions with others, and in \cite{GNEP-Shakarami-learning} the reward function is assumed to be linear. Closer to our scenario is \cite{GNEP-Alizadeh-learning}, which studied a centralized algorithm in a monotone game where the \textit{weighted empirical average} of the actions is shown to converge in expectation to a GNE.  In contrast, our GNE is limited to linear coupled constraints in strongly monotone games. For this scenario, we provide a \textit{distributed} stochastic learning scheme that converges in probability 1 to a GNE, without communication between players. To the best of our knowledge, our work is the first to study distributed and stochastic convergence to a GNE.

Our algorithm is a stochastic iterative approach to solving a constrained optimization problem using Lagrange multipliers. Convergence guarantees for stochastic augmented Lagrangian schemes are scarce. In \cite{constrained_1}, the continuous-time optimization problem was studied, and \cite{constrained_2} provides \textit{asymptotic} results for a variant of our algorithm. We present the convergence rate for our algorithm using novel analysis techniques.

Unlike mechanism design  \cite{heydaribeni2018distributed,deng2020robust}, our players are cooperative or myopic, and our manager does not
elicit any players' private information, but only observes the constraint violation. 

Our approach resembles a Stackelberg game \cite{maharjan2013dependable,fiez2020implicit,birmpas2020optimally}, where the manager is the leader and the players are the followers
(usually a single follower is considered). However, game-theoretic
algorithms converge to a NE that can be globally inefficient. Our work provides a mechanism
that provably shifts the unique NE to a point that satisfies target linear constraints that can force the NE to be efficient (e.g., first-order optimality condition for a global quadratic objective).

Our manager is learning with bandit feedback that is not stochastic
nor fully adversarial \cite{lattimore2020bandit,chandak2023equilibrium}. Instead, the rewards of this bandit are generated by the game dynamics. This ``game bandit''
introduces a new type of noise. The goal of the manager is to
steer the unique NE to a point that satisfies target linear constraints.
Hence, the feedback the manager needs concerns the behavior
at NE. However, the dynamics are not at NE every turn. In fact, by changing the prices, the manager perturbs the system which complicates and delays the convergence to the unique NE. Therefore, constraint violation feedback in this problem suffers from equilibrium noise which is
the distance of the dynamics from the time-varying NE.

Our work belongs to the literature on intervention
and control of games \cite{grammatico2017dynamic,parise2020analysis,parise2019variational,galeotti2020targeting,alpcan2009nash,mguni2019coordinating,brown2017optimal,ferguson2021effectiveness,bistritz2024gamekeeper}.
In \cite{tatarenko2014game} it was assumed that the manager knows
the reward functions of the players, and can compute the desired NE.
This is also the assumption in \cite{alpcan2009control,alpcan2009nash},
that considered the general approach of game-theoretic control. In
\cite{grammatico2017dynamic}, an aggregative game where the cost
is a linear function of the aggregated action was considered. It was
assumed that the manager can control the total loads the players observe,
which can deviate from the true total loads. In \cite{marden2009payoff,melo2011congestion,sandholm2007pigouvian},
a similar approach to ours was considered for discrete congestion
games, where the cost function of all players is the same convex
function of the number of players that share the chosen resource.
The congestion game considered there is a special case of monotone games that allow us to model far more general interactions and action
sets. For example, it allows for models when not only the number of players that share a resource matters but also their identity (e.g., geometric locations). Additionally, the pricing scheme in \cite{marden2009payoff,melo2011congestion,sandholm2007pigouvian}
requires the manager to know the reward functions. Control
of dynamic Markov potential games was considered in \cite{mguni2019coordinating},
where it was assumed that the manager knows the players' reward functions and can simulate the game.

Some works obtained convergence rates for contractive two time-scale stochastic approximation \cite{Konda_rate, Chandak-TTS-Opti}. These works assume that both the time-scales are fixed point iterations with contractive mappings. We consider a more general case where the slower timescale has a non-expansive mapping. 

Our convergence rate is limited by the non-expansive mapping in the slower time-scale, and not by the algorithm players use to update their actions. This non-expansive mapping arises due to the projection of the actions into the action sets in the faster time-scale. When the action sets are $\mathbb{R}^{d}$, our algorithm becomes a two-time-scale iteration with a contractive mapping in both time-scales, with a convergence rate of $O(t^{-1})$ \cite{Chandak-TTS-Opti}. This slowdown is a known challenge in the literature. While the optimal bound known for single time-scale fixed point iterations with contractive mappings is $O(t^{-1})$ \cite{chen2020finite}, the best known bound for iterations with non-expansive mappings is $O(t^{-1/3})$ \cite{Bravo}. Compared to $O(t^{-1/3})$, we attribute the slightly slower convergence rate to the two-time scale dynamics.

An earlier conference paper \cite{bistritz2021online} was limited to load balancing in resource allocation. This paper extends \cite{bistritz2021online} to linear constraints in strongly monotone games. The NE guarantees here are of the form $\bA\xstar(\balpha)=\bellstar$ instead of just $\sum_{n=1}^{N}\xx_{n}\left(\balphastar\right)=\bellstar$. Nevertheless, the results here are novel and independent of \cite{bistritz2021online}. Our convergence guarantees hold both almost surely and in $L_2$, whereas \cite{bistritz2021online} establishes only $L_2$ convergence. Unlike \cite{bistritz2021online}, we provide a finite time convergence error bound. Our novel analysis uses two time-scale stochastic approximation with non-expansive mappings.

\subsection{Outline and Notation}
The rest of the paper is organized as follows. Section \ref{sec:ProblemFormulation}  formulates our game control problem and details our assumptions. Subsection \ref{subsec:Applications} details two applications that fall under our general model.  Section \ref{sec:Algorithm} describes the game control algorithm and states our main results regarding its performance guarantees. Section \ref{sec:sketch} breaks the proofs of our main results into lemmas. Section \ref{sec:simulations} provides numerical simulations for the applications described earlier. Finally, Section \ref{sec:conclusion} concludes the paper and proposes future research directions. 

We use bold letters to denote vectors and matrices, and the standard game-theoretic notation where $\xx_{-n}$ is the vector of actions for all players except player $n$. $\|\cdot\|$ denotes the $\ell_2$-norm throughout the paper when applied on a vector, and the operator norm for matrices (under the $\ell_2$-norm for vectors). 

\section{Problem Formulation}\label{sec:ProblemFormulation}
Consider a set of players $\mathcal{N}=\{1,2,\ldots,N\}$, where player $n$ chooses an action $\xx_n=(x_n^1,\ldots,x_n^d)\in\XX_n$. Here $\XX_n\subset\RR^d_+$ is a convex and compact set such that $\bzero\in\XX_n$. Let $\XX=\XX_1\times\ldots\times\XX_N$. We use $\xx\in\XX$ to denote the concatenation of all action vectors. The utility for each player is defined as:
\begin{equation}\label{utility}
     u_{n}(\xx;\beta_{n}^{1},\ldots,\beta_{n}^{d})=r_{n}(\xx)-\sum_{i=1}^{d}\beta_{n}^{i}x_{n}^{i}.
\end{equation}
Here $r_n(\xx)$ is the reward function for player $n$. We assume that $r_n(\xx)$ is twice continuously differentiable for each $n$. The vector $\bbeta\in\RR^{Nd}$ denotes the `controlled coefficients' controlled by the game manager. We define the same vector notation for $\bbeta$ and $\bbeta_n$ as we do for $\xx$ and $\xx_n$, respectively. Such a linear term arises when introducing price per use of a resource, or when agents are cooperative and can  compute such a term since they know their own action. 

\begin{figure}[t]
\includegraphics[width=8.75cm,height=12cm,keepaspectratio]{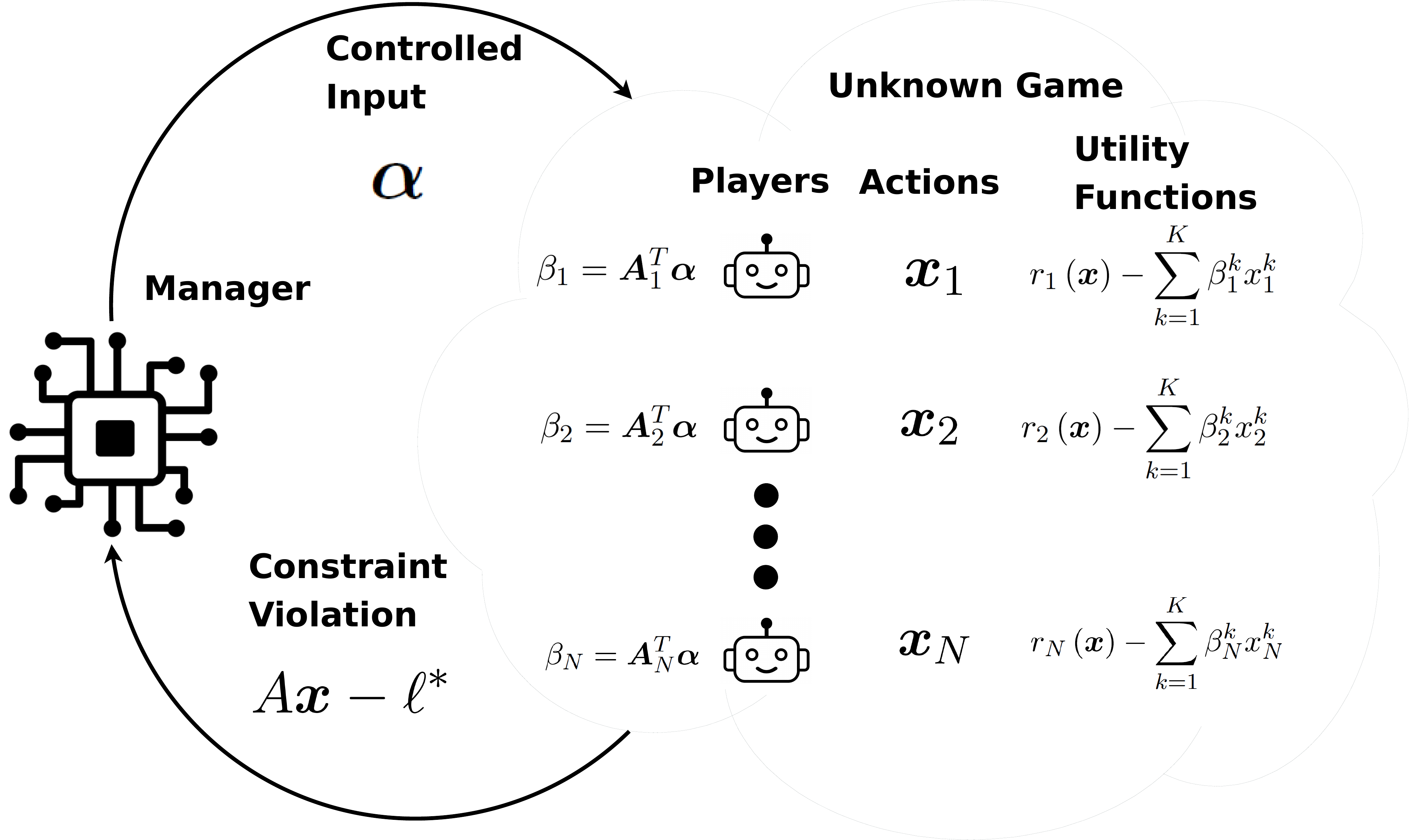}

\caption{\label{fig:system}Control of an unknown game}
\end{figure}

The key assumption is that the game is strongly monotone. For this, define the gradient operator $F(\cdot):\RR^{Nd}\mapsto\RR^{Nd}$:
\begin{equation}\label{F-defn}
    F(\xx)\coloneqq(\nabla_{\xx_1}r_1(\xx_1,\xx_{-1}),\ldots, \nabla_{\xx_N}r_N(\xx_N,\xx_{-N})).
\end{equation}

\begin{assumption}\label{assu-monotone}
There exists $\mu>0$ such that $\forall\xx,\yy\in\XX$, 
    \begin{equation}\label{F-mono}
      \langle\yy-\xx,F(\yy)-F(\xx)\rangle\leq -\mu\|\yy-\xx\|^2.
 \end{equation}

\end{assumption}
Thus, the game is strongly monotone with parameter $\mu>0$, or equivalently the function $-F(\cdot)$ is $\mu$-strongly monotone. Monotone games are a well-studied class of games \cite{VI-FP}. The linear term in \eqref{utility} does not affect the strongly monotone property of the game since the gradient vector becomes $F(\xx)-\bbeta$. In this strongly monotone game, there exists a unique pure NE $\xstar(\bbeta)$ corresponding to $\bbeta$, defined as:
\begin{definition}
    An action profile $(\xstar_n(\bbeta),\xstar_{-n}(\bbeta))$ is a pure Nash equilibrium (NE) corresponding to coefficients $\bbeta$, if $u_n(\xnstar(\bbeta),\xminusnstar(\bbeta);\bbeta)\geq u_n(\xx_n,\xminusnstar(\bbeta);\bbeta)$, for all $\xx_n\in\XX$ and all $n\in\mathcal{N}$. 
\end{definition}

We assume that the gradients of the players are Lipschitz continuous, which is a standard assumption \cite{smooth-ex1,smooth-ex2}.
\begin{assumption} 
There exists $L>0$ such that $\forall \xx,\yy\in\XX$,
\begin{equation}\label{F-lip}
     \|F(\yy)-F(\xx)\|\leq L\|\yy-\xx\|.
\end{equation}
\end{assumption}
 
Our game evolves in discrete time, so at time $t$ player $n$ plays $\xx_{n,t}$ given the controlled coefficients $\bbeta_{n,t}$. The players update their actions using gradient ascent and the controlled coefficients are adjusted based on the manager's objective (discussed below \eqref{manager-obj}). Since the game is strongly monotone, if the controlled coefficients are fixed at $\bbeta$ (i.e., uncontrolled game), the players will converge to the NE $\xstar(\bbeta)$ \cite{rosen1965existence}. 

In many scenarios, players do not fully know their reward functions but instead learn them on the fly based on the reward values they
receive or other data they collect. The gradient
can then be approximated from this data  \cite{leng2020learning}. These
learning schemes make the players' actions stochastic, which adds another
type of noise to our analysis.

\begin{assumption}
    Each player $n$ observes $\bg_{n,t-1}$ at time $t$, where $\bg_{n,t-1}=\nabla_{\xx_n}r_n(\xx_{t-1})+\MM_{n,t}$. Then $\bg_{t-1}=F(\xx_{t-1})+\MM_t$, where $\MM_t$ satisfies 
    \begin{equation}\label{Mart-bound}
        \EE[\MM_{t}\mid\FF_{t-1}]=0 \;\;\;\textrm{and}\;\;\;\EE[\|\MM_{t}\|^2\mid\FF_{t-1}]\leq M_c,
    \end{equation}
    for some $M_c>0$ and a filtration $\FF_t$ that summarizes the past:
    $$\FF_t=\sigma(\xx_{\tau}, \bbeta_{\tau}\mid\forall \tau\leq t).$$
\end{assumption}

The idea behind this assumption is that from the point of view of the players, the rewards are stochastic, i.e., $r_{n}(\xx)\triangleq r_{n}(\xx;\xi)$ for an unknown random variable $\xi$. Then the optimal estimator of the gradient of player $n$ from past observations is $\bg_{n,t}\triangleq\mathbb{E}\left\{\nabla_{\xx_n}r_n(\xx_{t})\mid\mathcal{F}_{t}\right\}$. This makes the estimation error $\MM_{n,t+1}=\bg_{n,t}-\nabla_{\xx_n}r_n(\xx_{t})$ a Martingale difference sequence, and $\left\Vert \MM_{n,t}\right\Vert$ is bounded if $\nabla_{\xx_n}r_n(\xx)$ is bounded. Gradient estimation from bandit feedback does not satisfy Assumption 3, which requires unbiased estimates.

The objective of the manager is to steer the players' NE towards a point that satisfies a set of $K$ linear constraints, which can together be represented in vector form as 
\begin{equation}\label{manager-obj}
    \bA\xx=\bellstar.
\end{equation} 
Here $\bA\in\RR^{K\times Nd}$ and $\bellstar\in\RR^{K}$. We assume that, at each time $t$, the manager only observes the constraint violation, i.e., the vector $(\bA\xx_{t-1}-\bellstar)$. These linear constraints can have different meanings for different applications. For resource allocation, the players' actions are how much to use of each resource and the manager wants to meet the target loads. When the global objective is a quadratic cost, the constraint violation vector is just the gradient of the global quadratic cost.

\subsection{Applications \label{subsec:Applications}}

\textcolor{black}{In this section, we show how our framework can be
applied in two different applications.}

\subsubsection{Quadratic Global Objective}
For a positive semi-definite matrix $\boldsymbol{G}\in\RR^{Nd\times Nd}$ and $\boldsymbol{\rho}\in\mathbb{R}^{Nd}$, the controller can set a global cost of the form $\Phi\left(\boldsymbol{x}\right)=\boldsymbol{x}^{T}\boldsymbol{G}\boldsymbol{x}+\boldsymbol{\rho}^T\boldsymbol{x}$. All the action profiles that are optima of $\Phi\left(\boldsymbol{x}\right)$ satisfy $\nabla\Phi\left(\boldsymbol{x}\right)=2\boldsymbol{G}\boldsymbol{x}+\boldsymbol{\rho}=0$. Hence, using our approach, we can set $\bA=2\boldsymbol{G}$ and $\bellstar=-\boldsymbol{\rho}$ to reach a global optimum.

\subsubsection{Weighted Resource Allocation Games}
In resource allocation,  high congestion on a single resource increases the operation costs for the manager, can lead to system failures, and incurs suboptimal performance for the users. Therefore, load balancing is a key control objective for the manager who supervises the resource allocation protocol. In the electricity grid, the manager can be the utility company running Demand Side Management (DSM) to cut the production costs of energy at peak hours. In wireless networks, the manager can be an access point coordinating a protocol that minimizes the power consumption of the devices and the interference in the network. Other examples are load balancing in data-centers \cite{bourke2001server} and control of parking resources \cite{dowling2017optimizing}. For distributed resource allocation, designing an efficient load balancing protocol becomes far more challenging. 

In weighted resource allocation, the player's action is how much they use from each of the $d$ resources. We have $K=d$ linear constraints, of the form  $\sum_{n=1}^Na_{n,i}^{(i)}x_n^{i}=\ell^*_{i}$ for the target load for each resource $i$. The manager sets $\bellstar$ to optimize application-specific objectives such as system efficiency, regulations, and operation costs. The controlled coefficients $\bbeta$ represent incentives to use a resource. The weights in $\bA$ are based on how each agent's use affects the resource. 

In wireless networks, the interference that the manager measures is a weighted sum of the transmission powers of the players, where the weights are the channel gains. In data centers, different users send jobs of different sizes to different servers (resources), so the load is a weighted sum of their actions (how many jobs to send where). A special case is resource allocation without weights, studied in \cite{bistritz2021online}. 

\section{Game Control Algorithm}\label{sec:Algorithm}
\begin{algorithm}[t]
\caption{\label{alg:DSM}Online Game Control}

\textbf{Initialization: }Let $\boldsymbol{x}_{0}\in\mathcal{X}$ and
$\balpha_{0}\in\RR^K$.

\textbf{Input: }Target vector $\bellstar$.

\textbf{For each turn $t\geq1$ do}
\begin{enumerate}
\item The manager observes the vector $\bA\xx_{t-1}-\bellstar$. 
\item Each player $n$ computes $\bbeta_{n,t-1}=\AT_{n}\balpha_{t-1}$ and updates its action using gradient ascent:
\begin{equation}\label{x-iter}
\boldsymbol{x}_{n,t}=\Pi_{\mathcal{X}_{n}}\left(\boldsymbol{x}_{n,t-1}+\eta_{t-1}\left(\bg_{n,t-1}-\bbeta_{n,t-1}\right)\right)
\end{equation}
where $\Pi_{\mathcal{X}_{n}}$ is the Euclidean projection into $\mathcal{X}_{n}$.

\item The manager updates the controlled input using
\begin{equation}\label{alpha-iter}
    \balpha_{t}=\balpha_{t-1}+\epsilon_{t-1}(\bA\xx_{t-1}-\bellstar).
\end{equation}
\item The manager broadcasts $\balpha_{t}$ to the agents.
\end{enumerate}
\textbf{End}
\end{algorithm}

Algorithm \ref{alg:DSM} details the combined scheme of the manager's
iteration and the players' behavior (modeled or designed). We discuss the key details of the algorithm in this section.

The manager controls the controlled coefficients $\bbeta$ in order to steer the dynamics towards a desirable point. 
But instead of directly performing computations on $\bbeta\in\RR^{Nd}$ and communicating it to agents, the manager uses the ``control input'' $\balpha\in\RR^K$, where $K$ is the number of linear constraints, which can be much smaller than $Nd$ depending on the application. After computing $\balpha_t$ at time $t$, the manager broadcasts $\balpha_t$ to the agents who then compute their $\bbeta_{n,t}$ using $\bbeta_{n,t}=\AT_n\balpha$. Here $\AT_n$ denotes the $d$ rows corresponding to player $n$ in the matrix $\AT$. Our algorithm requires each player to know only its corresponding $d$ rows in $\AT$. Note that $\bbeta=\AT\balpha$. Depending on the application, this row can represent local parameters known to the player (e.g., the size of the jobs the player sends to the server), or that the manager sends these parameters to players at the beginning of the algorithm.

The main challenge in guaranteeing that the dynamics converge to a point where the linear constraints are satisfied
is that the manager does not know the game, consisting of $\left\{ r_{n}\right\} $
and $\left\{ \mathcal{X}_{n}\right\} $. To overcome this, we propose
an online approach where the manager learns how to adjust
$\balpha_t$ to satisfy the linear constraints by using the iteration
in \eqref{alpha-iter} with the control step-size sequence $\left\{ \varepsilon_{t}\right\} $.
We only assume that the manager can observe the instantaneous
constraint violation $\bA\xx_{t-1}-\bellstar$ at the start of turn $t$. This feedback maintains basic privacy for the players since the manager does not know their reward functions and does not even observe their individual
actions.

Our algorithm is a two time-scale stochastic approximation iteration. In the faster time-scale, dictated by the stepsize $\{\eta_t\}$, the agents update their action using gradient play, taking a step towards the NE corresponding to $\balpha_t$. In the slower timescale, dictated by the stepsize $\{\epsilon_t\}$, the manager observes the constraint violation based on the players' current action and updates $\balpha_t$ to steer the player towards a desirable point.

Our players run SGD to optimize their rewards, and therefore are myopic and do not manipulate the manager using
dynamic strategies. This behavior is expected from selfish players in a large-scale network or from cooperative players. 

Selfish players are likely to maximize their rewards
using a common algorithm such as gradient ascent
\cite{mertikopoulos2019learning}.
In particular, when the players do not know their reward functions, small SGD steps allow for gradual learning, as opposed to best-response dynamics that can have large jumps. Nevertheless, our techniques can be applied to other distributed algorithms that are guaranteed to converge to NE. 

The behavior of \textit{cooperative} players is not something to model but to design \cite{bistritz2020cooperative}. The algorithm
is used as a distributed protocol programmed into the devices in the
network (e.g., WiFi). The challenge is then not to control the selfish
behavior of the players, but to guide them towards a more efficient NE they do not have enough information to find otherwise (i.e., players
do not know which $\balpha$ would satisfy the linear constraints at NE).
For this purpose, gradient play is an efficient distributed algorithm, that can lead to an efficient NE when combined with the manager's algorithm in \eqref{alpha-iter}. 

The manager's algorithm in \eqref{alpha-iter} is robust against strategic
manipulations of a small group of players since each player has a
negligible effect on $\bA\xx_{t-1}$, which
is its only way to impact \eqref{alpha-iter}. We conjecture that even a large group of players cannot manipulate \eqref{alpha-iter}
to increase their accumulated reward over time since any manipulated
change to $\balpha_{t}$ would be temporary and incurs
losses to the players. Moreover, truthful reporting is not an issue
since the players do not report anything, and the
manager cannot observe their actions. Instead, the manager only needs
to measure the instantaneous constraint violation.

\subsection{Performance Guarantees}\label{subsec:results}

The manager wants to converge to a suitable $\balpha$ such that the linear constraints $\bA\xx=\bellstar$ are satisfied at the NE corresponding to $\balpha$, i.e., converge to the set $$\Nopt=\{\balpha\mid \bA\xstar(\balpha)=\bellstar\},$$
where $\xstar(\balpha)$ denotes the NE corresponding to $\balpha$. 

We assume that Slater's condition \cite{convexbook} for our problem holds, which allows us to prove that $\Nopt$ is non-empty:
\begin{assumption}\label{assu-existence}
    The set of action profiles is of the form $\XX=\{\xx\mid h_i(\xx)\leq 0, i=1,\ldots,M\}$ for some $M$ where each $h_i(\xx)$ is continuously differentiable and convex and there exists $\xx\in\text{int}(\XX)$ (i.e., $h_i(\xx)<0$, $i=1,\ldots,M$) such that $\bA\xx=\bellstar$.
\end{assumption}

We also make the following standard assumptions on our step-size sequences \cite{Doan}. 
\begin{assumption}\label{assu-stepsize}
    The stepsizes $\{\eta_t\}$ and $\{\epsilon_t\}$ are of the form $\eta_t=1/(t+T_1)^\eta$ and $\epsilon_t=1/(t+T_2)^\epsilon$ for some $T_1,T_2>0$ and $0.5<\eta,\epsilon<1$ such that $\lim_{t\to\infty}\frac{\epsilon_t^2}{\eta_t^3}=0$.
\end{assumption}
The assumption that $0.5<\eta,\epsilon<1$ ensures that the stepsize sequences are monotonically decreasing and square-summable but not summable. The assumption that $\lim_{t\to\infty}\frac{\epsilon_t^2}{\eta_t^3}=0$ quantifies the required separation between the time-scales.


Our first result shows that $\balpha_t$ converges to $\Nopt$ and hence $\bA\xx_t$ converges to $\bellstar$. 
\begin{theorem}\label{thm:converge}
Under Assumptions 1-5, $\balpha_t$ converges to the set $\Nopt$, $\xx_t$ converges to $\xstar(\balpha_t)$, and with probability 1: $$\lim_{t\rightarrow\infty}\bA\xx_t=\bellstar.$$

\end{theorem}
The next theorem gives the convergence rate for our algorithm.      The bound holds for all $t>0$ if $T_1,T_2>T_0$. 
\begin{theorem}\label{thm:finite}
    Under Assumptions 1-5, there exists a $T_0>0$ and $C>0$ such that 
    $$\EE\left[\|\bA\xx_t-\bellstar\|^2\right]\leq C\left(\eta_t+\frac{1}{t\epsilon_t}\right),$$
    for all $t>T_0$.
\end{theorem}
\begin{remark}
    A smaller value of $\epsilon$ gives a better rate, but it is constrained by the assumptions that $\eta_t$ is square-summable ($\eta>0.5$) and $\lim_{t\rightarrow\infty}\frac{\epsilon_t^2}{\eta_t^3}=0$ ($2\epsilon>3\eta$). Let $\delta>0$. Then for $\eta=0.5+\frac{\delta}{3}$ and $\epsilon=0.75+\delta$ we obtain
     $$\EE\left[\|\bA\xx_t-\bellstar\|^2\right]\leq O\left(t^{-0.25+\delta}\right).$$
     where $\delta>0$ can be chosen to be arbitrarily small.

\end{remark}

\section{Convergence Analysis}\label{sec:sketch}

In this section, we detail our proof strategy for Theorem \ref{thm:converge} and \ref{thm:finite}. Throughout the proofs, we denote the gradient vector given the controlled input $\balpha$ by $F(\xx,\balpha)=F(\xx)-\AT\balpha$.

Our first lemma shows that the desired set $\Nopt$ is non-empty. The proof follows from the strong monotonicity of the game and Slater's conditions (Assumption \ref{assu-existence}). 

\begin{lemma}\label{lemma:existence}
    Under Assumptions \ref{assu-monotone} and \ref{assu-existence}, there exists $\balpha$ such that $\bA\xstar(\balpha)=\bellstar$, i.e., the set $\Nopt$ is non-empty.
\end{lemma}

\begin{proof}
Let $\mathcal{G}(\balpha)$ be our game with utility functions $u_n(\xx,\bbeta)$ and set of action profiles $\XX$ where $\bbeta=\AT\balpha$. Consider a modified game $\mathcal{G}'$ with utility functions $u'_n(\xx)=r_n(\xx)$ and set of action profiles $\{\xx\mid\xx\in\XX, \bA\xx=\bellstar\}$. Under the assumption that there exists $\xx$ such that $\bA\xx=\bellstar$ and $h_i(\xx)<0$ for all $i=1,\ldots,M$, Theorem 1 from \cite{rosen1965existence} states that there exists a NE $\hat{\xx}\in\XX$ of $\mathcal{G}'$. Further, the Karush-Kuhn-Tucker conditions (equations (3.3) and (3.6) from \cite{rosen1965existence}) state that there exist $\boldsymbol{\mu}\in\RR^M, \boldsymbol{\lambda}\in\RR^K$ such that 
$$F(\hat{\xx})+\sum_{i=1}^M\nabla h_i(\hat{\xx})\mu_i+\AT\boldsymbol{\lambda}=0 $$
and $$\bA\hat{\xx}=\bellstar \;\textrm{and}\; h_i(\hat{\xx})\leq0, \forall i\in\{1,2,\ldots,M\}.$$
Let $\balpha=-\boldsymbol{\lambda}$. Then the above equations can be written as
$$F(\hat{\xx}, \balpha)+\sum_{i=1}^M\nabla h_i(\hat{\xx})\mu_i=0 $$
and $$h_i(\hat{\xx})\leq0, \forall i\in\{1,2,\ldots,M\},$$
so $\hat{\xx}$ is a NE for the game $\mathcal{G}(\balpha)$, and hence $\xstar(\balpha)=\hat{\xx}$. We also know that $\bA\hat{\xx}=\bellstar$. Hence there exists $\balpha$ such that $\bA\xstar(\balpha)=\bellstar$, showing that $\Nopt$ is non-empty.


\end{proof}

A main technical challenge in our analysis is the convex and compact action sets. Indeed, with no constraints on the actions, and under certain assumptions on the matrix $\bA$, we can show that $\bA\xstar(\balpha)$ is strongly monotone instead of co-coercive (Lemma \ref{lemma:xstar-prop} (b)). This in turn shows that the mapping $g(\balpha)$ in Definition \ref{defn:g} is contractive. In this simplified case, both time-scales are fixed point iterations with contractive mappings and results from \cite{Chandak-TTS-Opti} can be applied. Therefore, the action sets call for a novel two time scale convergence analysis. 




\subsection{Non-expansive Mapping}\label{subsec:nonexpsketch}
In this subsection, we prove the non-expansiveness of the following mapping that is central to our analysis. 

\begin{definition}\label{defn:g}
    Define the map $g(\cdot):\RR^K\mapsto\RR^K$ as $$g(\balpha)=\balpha+\gamma(\bA\xstar(\balpha)-\bellstar),$$ where 
$\gamma=\min\{\mu/\|\bA\|^2,1\}$ is a constant chosen to satisfy $2\mu\gamma-\gamma^2\|\bA\|^2>0$. Note that a fixed point of $g(\balpha)$ will satisfy our objective of $\bA\xstar(\balpha)=\bellstar$. 
\end{definition}

To establish the non-expansiveness of $g$, we first need to show two important properties of the NE corresponding to each $\balpha\in\RR^K$. The first property is that $\xstar(\balpha)$ is Lipschitz continuous in $\balpha$ and the second property is that the weighted load vector $\bA\xstar(\balpha)$ is co-coercive in $\balpha$. The proofs for these two results follow from the strongly monotone nature of the game by using properties of variational inequalities (VI). 

\begin{lemma}\label{lemma:xstar-prop}
    Let $\xstar(\balpha)$ be the NE given $\balpha\in\RR^K$. Then, 
    \begin{enumerate}[label=(\alph*)]
        \item There exists a constant $L_0$ such that for all $\balpha_1,\balpha_2\in\RR^K$, 
        \begin{equation}\label{xstar-lip}
            \|\xstar(\balpha_2)-\xstar(\balpha_1)\|\leq L_0\|\balpha_2-\balpha_1\|.
        \end{equation}
        \item For all $\balpha_1,\balpha_2\in\RR^K$, 
        \begin{align}\label{xstar-coco}
            &\langle \bA\xstar(\balpha_2)-\bA\xstar(\balpha_1), \balpha_2-\balpha_1\rangle\nonumber\\
            &\;\;\;\;\;\;\;\;\;\;\;\;\;\;\;\;\;\;\;\;\;\leq-\mu\|\xstar(\balpha_2)-\xstar(\balpha_1)\|^2,
        \end{align}
        where $\mu$ is the monotonicity parameter of the game.
    \end{enumerate}
\end{lemma}

\begin{proof}
For a fixed $\balpha$, $-F(\xx,\balpha)$ is strongly monotone on $\XX$ and is Lipschitz continuous by Assumptions 1 and 2 respectively. Then we state two properties of the NE $\xstar(\balpha)$. First, \cite[Proposition 1.4.2]{VI-FP} states that the NE $\xstar(\balpha)$ corresponding to $\balpha$ satisfies the variational inequality 
\begin{equation}\label{VI-prop-1}
    \left\langle \xx-\xstar(\balpha),F\left(\xstar(\balpha),\balpha\right)\right\rangle \leq0
\end{equation} for all $\xx\in\mathcal{X}$. Second, for a strongly monotone function, \cite[Theorem 2.3.3]{VI-FP} states that for all $\xx\in\XX$, for some $L'$   
\begin{equation}\label{VI-prop-2}
    \|\xx-\xstar(\balpha)\|\leq L'\|\xx-\Pi_\XX(\xx+F(\xx,\balpha))\|.
\end{equation}
Hence for $\balpha_1,\balpha_2\in\RR^K$, we have
\begin{align*}
    &\|\xstar(\balpha_2)-\xstar(\balpha_1)\|\\
    &\stackrel{(a)}{\leq} L'\|\xstar(\balpha_2)-\Pi_\XX(\xstar(\balpha_2)+F(\xstar(\balpha_2),\balpha_1))\|\\
    &\stackrel{(b)}{=}L'\Big\|\Pi_\XX(\xstar(\balpha_2)+F(\xstar(\balpha_2),\balpha_2))\\
    &\;\;\;\;\;\;\;\;\;-\Pi_\XX(\xstar(\balpha_2)+F(\xstar(\balpha_2),\balpha_1))\Big\|\\
    &\leq L'\|F(\xstar(\balpha_2),\balpha_2)-F(\xstar(\balpha_2),\balpha_1)\|\\
    &=L'\|\AT(\balpha_2-\balpha_1)\|\\
    &\leq L_0\|\balpha_2-\balpha_1\|,
\end{align*}
where $L_0=L'\|\AT\|$. Here inequality (a) follows from \eqref{VI-prop-2} and equality (b) follows from \cite[Proposition 1.5.8]{VI-FP} which states that the solution of the VI \eqref{VI-prop-1} satisfies
$$\xstar(\balpha_2)=\Pi_\XX(\xstar(\balpha_2)+F(\xstar(\balpha_2),\balpha_2)).$$
This completes the proof for part (a) of Lemma \ref{lemma:xstar-prop}.

Now let $\balpha_1,\balpha_2\in\RR^K$. Then for $x=\xstar(\balpha_2)$ in \eqref{VI-prop-1},
\begin{equation*}
\left\langle \xstar(\balpha_2)-\xstar(\balpha_1),F\left(\xstar(\balpha_1),\balpha_1\right)\right\rangle \leq0.
\end{equation*}
Similarly, we also have 
\begin{equation*}
\left\langle \xstar(\balpha_1)-\xstar(\balpha_2),F\left(\xstar(\balpha_2),\balpha_2\right)\right\rangle \leq0.
\end{equation*}
By combining the above two inequalities we obtain
\begin{equation}\label{random-eqn}
\left\langle \xstar(\balpha_2)-\xstar(\balpha_1),F\left(\xstar(\balpha_2),\balpha_2\right)-F\left(\xstar(\balpha_1),\balpha_1\right)\right\rangle \geq0.
\end{equation}

Then 
\begin{align*}
    &-\mu\|\xstar(\balpha_2)-\xstar(\balpha_1)\|^2\\
    &\stackrel{(a)}{\geq} \langle   \xstar(\balpha_2)-\xstar(\balpha_1),F\left(\xstar(\balpha_2),\balpha_1\right)-F\left(\xstar(\balpha_1),\balpha_1\right)\rangle\\
    &= \langle   \xstar(\balpha_2)-\xstar(\balpha_1),F\left(\xstar(\balpha_2),\balpha_2\right)-F\left(\xstar(\balpha_1),\balpha_1\right)\rangle\\
    &\;\;\;+\langle   \xstar(\balpha_2)-\xstar(\balpha_1),\AT\balpha_2-\AT\balpha_1\rangle\\
    &\stackrel{(b)}{\geq} \langle   \xstar(\balpha_2)-\xstar(\balpha_1),\AT\balpha_2-\AT\balpha_1\rangle\\
    &=(\xstar(\balpha_2)-\xstar(\balpha_1))^\intercal\AT(\balpha_2-\balpha_1)\\
    &=(\bA\xstar(\balpha_2)-\bA\xstar(\balpha_1))^\intercal(\balpha_2-\balpha_1)\\
    &=\left\langle\bA\xstar(\balpha_2)-\bA\xstar(\balpha_1),\balpha_2-\balpha_1\right\rangle.
\end{align*}
Here inequality (a) uses that $-F(\xx,\balpha)$ is strongly monotone in $\xx$ with parameter $\mu$ and (b) uses \eqref{random-eqn}. 
\end{proof}

The next lemma shows that the map $g(\cdot)$ is non-expansive.
\begin{lemma}\label{lemma:g-prop}
    For any $\balpha_1,\balpha_2\in\RR^K$, the function $g(\cdot)$ satisfies, 
        \begin{equation}\label{non-exp}
            \|g(\balpha_2)-g(\balpha_1)\|\leq \|\balpha_2-\balpha_1\|,
        \end{equation}
        i.e., the map $g(\cdot)$ is non-expansive.
\end{lemma}

\begin{proof}
For any $\balpha_2,\balpha_1\in\RR^K$, 
    \begin{align*}
        &\|g(\balpha_2)-g(\balpha_1)\|^2\\
        &=\|\balpha_2+\gamma\bA\xstar(\balpha_2)-\balpha_1-\gamma\bA\xstar(\balpha_1)\|^2\\
        &=\|\balpha_2-\balpha_1\|^2+\gamma^2\left\|\bA\xstar(\balpha_2)-\bA\xstar(\balpha_1)\right\|^2\\
        &\;\;\;+2\gamma\left\langle \bA\xstar(\balpha_2)-\bA\xstar(\balpha_1),\balpha_2-\balpha_1\right\rangle \\
        &\stackrel{(a)}\leq \|\balpha_2-\balpha_1\|^2+\gamma^2\|\bA\|^2\|\xstar(\balpha_2)-\xstar(\balpha_1)\|^2\\
        &\;\;\;-2\mu\gamma\|\xstar(\balpha_2)-\xstar(\balpha_1)\|^2\\
        &= \|\balpha_2-\balpha_1\|^2-(2\mu\gamma-\gamma^2\|\bA\|^2)\|\xstar(\balpha_2)-\xstar(\balpha_1)\|^2\\
        &\stackrel{(b)}{\leq} \|\balpha_2-\balpha_1\|^2.
    \end{align*}
    Here inequality (a) follows from part (b) of Lemma \ref{lemma:xstar-prop} and inequality (b) follows from the assumption on $\gamma$ that $2\mu\gamma>\gamma^2\|\bA\|^2$. Hence $g(\cdot)$ is non-expansive. 
\end{proof}

The next lemma shows that the iterate $\balpha_t$ has bounded expectation for all $t$. This lemma is required both for proving the asymptotic convergence result in Theorem \ref{thm:converge}, which requires almost-sure boundedness of iterates and to obtain a better convergence rate in Theorem \ref{thm:finite}.
\begin{lemma}\label{lemma:bounded}
    Under Assumptions 1-5,  
    $$\EE[\|\balpha_t-\balphastar\|^2]\leq C_0,$$
    for some $C_0>0$, $\balphastar\in\Nopt$ and all $t>0$.
    \end{lemma}

\begin{proof}
    Let $\balphastar\in\Nopt$. We first give a recursive bound on the term $\EE[\|\xx_{t+1}-\xstar(\balpha_{t+1})\|^2]$. For this, note that
\begin{subequations}\label{split1}
\begin{align}
    &\|\xx_{t+1}-\xstar(\balpha_{t+1})\|^2\nonumber\\
    &=\|\Pi_\XX(\xx_t+\eta_t(F(\xx_t)-\ATalpha_t+\MM_{t+1}))-\xstar(\balpha_{t+1})\|^2\nonumber\\
    &\stackrel{(a)}{=}\|\Pi_\XX(\xx_t+\eta_t(F(\xx_t)-\ATalpha_t+\MM_{t+1}))\nonumber\\
    &\;\;\;\;\;\;\;\;\;\;\;\;\;\;\;\;\;\;\;\;\;\;\;\;\;\;\;\;\;\;\;\;\;\;\;\;\;\;\;\;\;\;\;\;\;\;\;\;\;-\Pi_\XX(\xstar(\balpha_{t+1}))\|^2\nonumber\\
    &\stackrel{(b)}{\leq}\|\xx_t+\eta_t(F(\xx_t)-\ATalpha_t+\MM_{t+1})-\xstar(\balpha_{t+1})\|^2\nonumber\\
    &=\|\xx_t-\xstar(\balpha_t)+\eta_t(F(\xx_t)-\ATalpha_t+\MM_{t+1})\nonumber\\
    &\;\;+\xstar(\balpha_t)-\xstar(\balpha_{t+1})\|^2\nonumber\\
    &=\|\xx_t-\xstar(\balpha_t)+\eta_tF(\xx_t,\balpha_t)\|^2\label{split11}\\
    &\;\;+ \|\eta_t\MM_{t+1}+\xstar(\balpha_t)-\xstar(\balpha_{t+1})\|^2\label{split12}\\
    &\;\;+ 2\big\langle\xx_t-\xstar(\balpha_t)+\eta_tF(\xx_t,\balpha_t),\nonumber \\
    &\;\;\;\;\;\;\;\;\;\;\;\;\;\;\;\;\;\;\eta_t\MM_{t+1}+\xstar(\balpha_t)-\xstar(\balpha_{t+1})\big\rangle.\label{split13}
\end{align}    
\end{subequations}
Here equality (a) uses that $\xstar(\balpha_{t+1})\in\XX$ and inequality (b) follows since the projection operator $\Pi_\XX(\cdot)$ for any convex set $\XX$ is non-expansive.

We first simplify the term \eqref{split11}:
\begin{align*}
    &\|\xx_t-\xstar(\balpha_t)+\eta_tF(\xx_t,\balpha_t)\|^2\\
    &=\|\xx_t-\xstar(\balpha_t)\|^2+\eta_t^2\|F(\xx_t,\balpha_t)\|^2\\
    &\;\;+ 2\eta_t\langle\xx_t-\xstar(\balpha_t),F(\xx_t,\balpha_t)\rangle.
\end{align*}
The last term is bounded as follows.
\begin{align*}
    &\langle\xx_t-\xstar(\balpha_t), F(\xx_t,\balpha_t)\rangle\\
    &\leq \langle\xx_t-\xstar(\balpha_t), F(\xx_t,\balpha_t)-F(\xstar(\balpha_t),\balpha_t)\rangle\\
    &\leq -\mu\|\xx_t-\xstar(\balpha_t)\|^2.
\end{align*}
Here the first inequality follows since $\langle\xx-\xstar(\balpha),F(\xstar(\balpha),\balpha)\rangle\leq 0$, for all $\balpha\in\RR^K$ and $\xx\in\XX$, since $\xstar(\balpha)$ is a NE (Proposition 1.4.2 in \cite{VI-FP}) and the second inequality follows from the strong monotonicity property of the game. We also have that 
\begin{align*}
    &\|F(\xx_t,\balpha_t)\|^2=\|F(\xx_t)-\AT\balpha_t\|^2\\
    &=\|F(\xx_t)-\AT\balphastar-\AT\balpha_t+\AT\balphastar\|^2\\
    &\leq 2\|F(\xx_t)-\AT\balphastar\|^2+2\|\AT\|^2\|\balpha_t-\balphastar\|^2\\
    &\leq c_1^2+c_2^2\|\balpha_t-\balphastar\|^2.
\end{align*}
where the last inequality follows for some $c_1$ and $c_2$ as $\|F(\xx_t)-\AT\balphastar\|$ is bounded by some constant as $\xx_t\in\XX$, where $\XX$ is compact. Coming back to \eqref{split11}, we have
\begin{align*}
    &\|\xx_t-\xstar(\balpha_t)+\eta_tF(\xx_t,\balpha_t)\|^2\\
    &\leq (1-2\mu\eta_t)\|\xx_t-\xstar(\balpha_t)\|^2+\eta_t^2c_1^2+\eta_t^2c_2^2\|\balpha_t-\balphastar\|^2.
\end{align*}
For \eqref{split12}, note that
\begin{align*}
    &\EE[\|\eta_t\MM_{t+1}+\xstar(\balpha_t)-\xstar(\balpha_{t+1})\|^2\mid\FF_t]\\
    &= \eta_t^2\EE[\|\MM_{t+1}\|^2\mid\FF_t]+\EE[\|\xstar(\balpha_t)-\xstar(\balpha_{t+1})\|^2\mid\FF_t]\\
    &\;\;+ 2\eta_t\EE[\langle\MM_{t+1},\xstar(\balpha_t)-\xstar(\balpha_{t+1})\rangle\mid\FF_t]\\
    &\stackrel{(a)}{=}\eta_t^2\EE[\|\MM_{t+1}\|^2\mid\FF_t]+\|\xstar(\balpha_t)-\xstar(\balpha_{t+1})\|^2\\
    &\stackrel{(b)}{\leq} \eta_t^2M_c + L_0^2\|\balpha_t-\balpha_{t+1}\|^2= \eta_t^2M_c+ L_0^2\epsilon_t^2\left\|\bA\xx_t-\bellstar\right\|^2\\
    &\stackrel{(c)}{\leq} \eta_t^2M_c+ \epsilon_t^2c_3^2.
\end{align*}
Here (a) uses $\EE[\langle\MM_{t+1},\xstar(\balpha_t)-\xstar(\balpha_{t+1})\rangle\mid\FF_t]=0$, since $\alpha_{t+1}$ is $\FF_t$ measurable and $\MM_{t+1}$ is a Martingale difference sequence. Inequality (b) follows from Lemma \ref{lemma:xstar-prop} part (a), and (c) follows from $\xx_t\in\XX$ for some constant $c_3$. 

Finally, for the third term \eqref{split13}:
\begin{align*}
    &2\EE\big[\big\langle\xx_t-\xstar(\balpha_t)+\eta_tF(\xx_t,\balpha_t),\\
    &\;\;\;\;\;\;\;\;\;\;\;\;\;\;\;\;\eta_t\MM_{t+1}+\xstar(\balpha_t)-\xstar(\balpha_{t+1})\big\rangle\mid\FF_t\big]\\
    &= 2\big\langle\xx_t-\xstar(\balpha_t)+\eta_tF(\xx_t,\balpha_t),\xstar(\balpha_t)-\xstar(\balpha_{t+1})\big\rangle\\
    &\stackrel{(a)}{\leq} 2\|\xx_t-\xstar(\balpha_t)+\eta_tF(\xx_t,\balpha_t)\|L_0\|\balpha_t-\balpha_{t+1}\|\\
    &\stackrel{}{\leq}2\|\xx_t-\xstar(\balpha_t)\|\epsilon_tL_0\left\|\bA\xx_t-\bellstar\right\|\\
    &\;\;+2\eta_t\|F(\xx_t,\balpha_t)\|\epsilon_tL_0\|\bA\xx_t-\bellstar\|\\
    &\stackrel{(b)}{\leq}2\|\xx_t-\xstar(\balpha_t)\|\epsilon_tc_3+2\eta_t\epsilon_t\|F(\xx_t,\balpha_t)\|c_3.
\end{align*}
Here the first equality uses that $\MM_{t+1}$ is a martingale difference sequence and $\xx_t$, $\balpha_t$ and $\balpha_{t+1}$ are $\FF_t$ measurable. Inequality (a) follows from the Cauchy-Schwarz inequality and the Lipschitz continuity of the map $\xstar(\cdot)$ (Lemma \ref{lemma:xstar-prop}) and (b) uses boundedness of $\xx_t$ and similar to previous expansions.

To bound the term $2\|\xx_t-\xstar(\balpha_t)\|\epsilon_tc_3$, we use $(2ab\leq a^2+b^2)$ where $a=\sqrt{\mu\eta_t\|\xx_t-\xstar(\balpha_t)\|^2/2}$ and $b=\sqrt{(\epsilon_t^2/\eta_t)(2c_3^2/\mu)}$. To bound the term $2\eta_t\epsilon_t\|F(\xx_t,\balpha_t)\|c_3$, we use the same inequality with $a=\eta_t\|F(\xx_t,\balpha_t)\|$ and $b=\epsilon_tc_3$. This yields
\begin{align*}
    &2\EE\big[\big\langle\xx_t-\xstar(\balpha_t)+\eta_tF(\xx_t,\balpha_t),\\
    &\;\;\;\;\;\;\;\;\;\;\;\;\;\;\;\;\eta_t\MM_{t+1}+\xstar(\balpha_t)-\xstar(\balpha_{t+1})\big\rangle\mid\FF_t\big]\\
    &\leq \frac{\mu}{2}\eta_t\|\xx_t-\xstar(\balpha_t)\|^2+\frac{\epsilon_t^2}{\eta_t}2c_3^2/\mu+\eta_t^2\|F(\xx_t,\balpha_t)\|^2+\epsilon_t^2c_3^2\\
    &\leq \frac{\mu}{2}\eta_t\|\xx_t-\xstar(\balpha_t)\|^2+\frac{\epsilon_t^2}{\eta_t}2c_3^2/\mu\\
    &\;\;\;+\eta_t^2c_1^2+\eta_t^2c_2^2\|\balpha_t-\balphastar\|^2+\epsilon_t^2c_3^2\\
\end{align*}

We have now bounded all three terms from \eqref{split1}. Combining them after taking expectation gives us:
\begin{align*}
    &\EE[\|\xx_{t+1}-\xstar(\balpha_{t+1})\|^2]\nonumber\\
    &\leq \left(1-\frac{3}{2}\mu\eta_t\right)\EE[\|\xx_t-\xstar(\balpha_t)\|^2]+\eta_t^2(2c_1^2+M_c)\\
    &\;\;\;+2\epsilon_t^2c_3^2+\frac{\epsilon_t^2}{\eta_t}2c_3^2/\mu+2\eta_t^2c_2^2\EE\left[\|\balpha_t-\balphastar\|^2\right].
\end{align*}
Our assumption that $\lim_{t\to\infty}\epsilon_t^2/\eta_t^3\rightarrow0$ implies that $\epsilon_t^2\leq C_1\eta_t^3$ for some  $C_1>0$ and $t>T_0$. Then for all $t>T_0$:
\begin{align}\label{last-bound-x}
    \EE[\|\xx_{t+1}-\xstar(\balpha_{t+1})\|^2]&\leq (1-1.5\mu\eta_t)\EE[\|\xx_t-\xstar(\balpha_t)\|^2]\nonumber\\
    +&c_4\eta_t^2+2\eta_t^2c_2^2\EE\left[\|\balpha_t-\balphastar\|^2\right],
\end{align}
for $c_4=2c_1^2+M_c+2C_1c_3^2+2C_1c_3^2/\mu$.

We next give a recursive bound on $\EE[\|\balpha_{t+1}-\balphastar\|^2]$. Let $\epsilon_t'=\epsilon_t/\gamma$. Then iteration \eqref{alpha-iter} can be written as $$\balpha_{t+1}=\balpha_t+\epsilon_t'(g(\balpha_t)-\balpha_t)+\epsilon_t'\bA\gamma(\xx_t-\xstar(\balpha_t)),$$
where $g(\balpha)=\balpha+\gamma(\bA\xstar(\balpha)-\bellstar)$.
    For ease of notation define the sequences $p_t=g(\balpha_t)-\balpha_t$ and $q_t=\gamma \bA(\xx_t-\xstar(\balpha_t))$. Let $\balphastar\in\Nopt$ be a fixed point of $g(\cdot)$. Then
    \begin{align*}
        &\|\balpha_{t+1}-\balphastar\|^2\\
        &=\|\balpha_t+\epsilon_t'(g(\balpha_t)-\balpha_t+q_t)-\balphastar\|^2\\
        &=\|(1-\epsilon_t')(\balpha_t-\balphastar)+\epsilon'_t(g(\balpha_t)-\balphastar+q_t)\|^2.
    \end{align*}
    Using Corollary 2.14 from \cite{convexbook}, we obtain
    \begin{align}\label{split2}
        \|\balpha_{t+1}-\balphastar\|^2&\leq(1-\epsilon_t')\|\balpha_t-\balphastar\|^2\nonumber\\
        &\;\;\;+\epsilon_t'\|g(\balpha_t)-\balphastar+q_t\|^2\nonumber\\
        &\;\;\;-\epsilon_t'(1-\epsilon_t')\|p_t+q_t\|^2.
    \end{align}
    We bound the second term as follows:
    \begin{align*}
        &\|g(\balpha_t)-\balphastar+q_t\|^2\\
        &=\|g(\balpha_t)-\balphastar\|^2+\|q_t\|^2+2\langle g(\balpha_t)-\balphastar,q_t\rangle\\
        &\leq \|\balpha_t-\balphastar\|^2+\|q_t\|^2+2\| g(\balpha_t)-\balphastar\|\|q_t\|\\
        &\leq \|\balpha_t-\balphastar\|^2+\|q_t\|^2+2\| \balpha_t-\balphastar\|\|q_t\|
    \end{align*}
    Here the first inequality stems from the non-expansive nature of map $g(\cdot)$ and the final inequality uses the fact that $g(\balphastar)=\balphastar$ along with the non-expansive nature of $g(\cdot)$. Now
    \begin{align*}
        &2\epsilon_t'\|\balpha_t-\balphastar\|\|q_t\|\leq 2\epsilon_t\|\bA\|\|\balpha_t-\balphastar\|\|\xx_t-\xstar(\balpha_t)\|\\
        &\leq \mu\eta_t/2\|\xx_t-\xstar(\balpha_t)\|^2+\frac{2\|\bA\|^2\epsilon_t^2}{\mu\eta_t}\|\balpha_t-\balphastar\|^2.
    \end{align*}
    Then returning to \eqref{split2}, 
    \begin{align*}
        &\|\balpha_{t+1}-\balphastar\|^2\leq \|\balpha_t-\balphastar\|^2+\epsilon_t'\gamma^2\|\bA\|^2\|\xx_t-\xstar(\balpha_t)\|^2\\
        &\;\;\;+\mu\eta_t/2\|\xx_t-\xstar(\balpha_t)\|^2+\frac{2\|\bA\|^2\epsilon_t^2}{\mu\eta_t}\|\balpha_t-\balphastar\|^2\\
        &\;\;\;-\epsilon_t'(1-\epsilon_t')\|p_t+q_t\|^2.
    \end{align*}
    Again using our assumption that $\epsilon_t^2\leq C_1\eta_t^3$ for some appropriate $C_1>0$ and $t>T_0$, and taking expectation we have 
    \begin{align}\label{last-bound-alpha}
        &\EE\left[\|\balpha_{t+1}-\balphastar\|^2\right]\nonumber\\
        &\leq \EE\left[\|\balpha_{t}-\balphastar\|^2\right]+\mu\eta_t\EE\left[\|\xx_t-\xstar(\balpha_t)\|^2\right]\nonumber\\
        &\;\;\;+2\eta_t^2c_2^2\EE\left[\|\balpha_{t}-\balphastar\|^2\right]-\epsilon_t'(1-\epsilon_t')\EE[\|p_t+q_t\|^2].
    \end{align}

    Adding \eqref{last-bound-x} and \eqref{last-bound-alpha}, and dropping  negative terms, we get
    \begin{align*}
        &\EE\left[\|\xx_{t+1}-\xstar(\balpha_{t+1})\|^2+\|\balpha_{t+1}-\balphastar\|^2\right]\\
        &\leq (1+4c_2^2\eta_t^2)\EE\left[\|\xx_t-\xstar(\balpha_t)\|^2+\|\balpha_{t}-\balphastar\|^2\right]+c_4\eta_t^2.
    \end{align*}
    Expanding the recursion for $t>T_0$ gives us
    \begin{align*}
        &\EE\left[\|\xx_{t}-\xstar(\balpha_{t})\|^2+\|\balpha_{t+1}-\balphastar\|^2\right]\\
        &\leq \EE\left[\|\xx_{T_0}-\xstar(\balpha_{T_0})\|^2+\|\balpha_{T_0}-\balphastar\|^2\right]\prod_{i=T_0}^{t-1}(1+4c_2^2\eta_i^2)\\
        &\;\;\;+c_4\sum_{i=T_0}^{t-1}\eta_i^2\prod_{j=i+1}^{t-1}(1+4c_2^2\eta_j^2)\\
        &\leq \EE\left[\|\xx_{T_0}-\xstar(\balpha_{T_0})\|^2+\|\balpha_{T_0}-\balphastar\|^2\right]e^{\left(4c_2^2\sum_{i=T_0}^{t-1}\eta_i^2\right)}\\
        &\;\;\;+c_4\sum_{i=T_0}^{t-1}\eta_i^2e^{\left(4c_2^2\sum_{j=i+1}^{t-1}\eta_j^2\right)}.
    \end{align*}
Here the final inequality uses that $(1+x)\leq e^x$ for any $x\in\RR$. Now we use our assumption that $\eta_t$ is square-summable. Hence, $\exp(4c_2^2\sum_{i=T_0}^{t-1}\eta_i^2)$ and $\sum_{i=T_0}^{t-1}\eta_i^2\exp\left(4c_2^2\sum_{j=i+1}^{t-1}\eta_j^2\right)$ are bounded by constants. Finally, $T_0$ is a constant dependent only on the stepsize sequence, and hence $\|\xx_{T_0}-\xstar(\balpha_{T_0})\|^2+\|\balpha_{T_0}-\balphastar\|^2$ can be bounded by a constant using the discrete Gronwall's inequality \cite{VSBbook}. This finally gives us that for some constant $C_0$
$$\EE\left[\|\xx_{t}-\xstar(\balpha_{t})\|^2+\|\balpha_{t+1}-\balphastar\|^2\right]\leq C_0,$$
which completes the proof for Lemma \ref{lemma:bounded}.
\end{proof}

Lemmas 4-6 are sufficient to prove Theorem \ref{thm:converge} (see the Appendix). We use two time-scale stochastic approximation to show that $(\xx_t,\balpha_t)\to(\xstar(\balpha_t),\balpha_t)$. Then we can analyze the slower iteration of the control updates \eqref{alpha-iter} in the limit $\xx_t\to\xstar(\balpha_t)$. We rewrite \eqref{alpha-iter} as a fixed point iteration with the mapping $g(\cdot)$. Since $g(\cdot)$ is non-expansive, this iteration is a version of the inexact KM algorithm \cite{convexbook}. We study its ODE limit and show convergence to the set of fixed points.

\subsection{Finite Time Convergence Guarantees}\label{subsec:finitesketch}
Our finite time convergence guarantees also use the non-expansiveness on $g$, but require a more nuanced argument that bounds the error as follows
\begin{align*}
    \EE[\|\bA\xx_t-\bellstar\|^2]&\leq 2\EE[\|\bA(\xx_t-\xstar(\balpha_t))\|^2]\\
    &\;\;\;\;\;\;\;+2\EE[\|\bA\xstar(\balpha_t)-\bellstar\|^2].
\end{align*}
We first bound the mean square distance between $\xx_t$ and $\xstar(\balpha_t)$, i.e., the NE corresponding to $\balpha_t$. The key to proving this lemma is the strong monotonicity of the game which gives a recursive bound on $\EE[\|\xx_t-\xstar(\balpha_t)\|^2]$. This recursive bound gives a rate of $O(\eta_t)$ when solved. 
\begin{lemma}\label{lemma:inner}
    Under Assumptions 1-5, there exist $C'>0$ and $T_0$ such that  for all $t>T_0$:
    $$\EE[\|\xx_t-\xstar(\balpha_t)\|^2]\leq C'\eta_t.$$
\end{lemma}

\begin{proof}
Returning to the the recursive bound \eqref{last-bound-x}  from the proof for Lemma \ref{lemma:bounded}, and using Lemma \ref{lemma:bounded} that shows that
    $\EE\left[\|\balpha_t-\balphastar\|^2\right]\leq C_0$, we now obtain the recursion
    \begin{align*}
    \EE[\|\xx_{t+1}-\xstar(\balpha_{t+1})\|^2]&\leq (1-1.5\mu\eta_t)\EE[\|\xx_t-\xstar(\balpha_t)\|^2]\nonumber\\
    &+\eta_t^2(c_4+2c_2^2C_0).
\end{align*}
For $\eta_t=(t+T_1)^{-\eta}$ with $0<\eta<1$, this yields \cite{chen2020finite}
$$\EE[\|\xx_t-\xstar(\balpha_t)\|^2]\leq C'\eta_t,$$
for $t>T_0$ for some constant $C'$.
\end{proof}

The next lemma bounds the mean squared distance between $\bA\xstar(\balpha_t)$ and $\bellstar$. Note that we need to bound $\EE[\|g(\balpha_t)-\balpha_t\|^2]$ as $g(\balpha)-\balpha=\gamma(\bA\xstar(\balpha)-\bellstar)$. We prove the bound in the following lemma using Lemma \ref{lemma:inner} and a technique similar to the bound for (inexact) KM algorithm. If we could assume that $\epsilon_t>\epsilon^*$ for some $\epsilon^*$, then $\EE[\|g(\balpha_t)-\balpha_t\|^2]$ is known to decay at the rate $o(1/t)$ \cite{nonexp-1}. Since this assumption is not true in our case, we use the assumptions on the stepsize sequences to get a convergence rate of $O(1/(t\epsilon_t))$ instead.

\begin{lemma}\label{lemma:outer}
    Let Assumptions 1-5 be satisfied. Then, there exist $C''>0$ and $T_0>0$ such that for all $t>T_0$.
    $$\EE[\|\bA\xstar(\balpha_t)-\bellstar\|^2]\leq \frac{C''}{t\epsilon_t}.$$
\end{lemma}

\begin{proof}
   Define $p_t=g(\balpha_t)-\balpha_t$, $q_t=\gamma \bA(\xx_t-\xstar(\balpha_t))$ and $\epsilon_t'=\epsilon_t/\gamma$. Returning to the recursive bound \eqref{last-bound-alpha} from the proof of Lemma \ref{lemma:bounded}, and now using Lemma \ref{lemma:bounded} and Lemma \ref{lemma:inner}, 
    \begin{align}\label{split22}
        &\EE\left[\|\balpha_{t+1}-\balphastar\|^2\right]\nonumber\nonumber\\
        &\leq \EE\left[\|\balpha_{t}-\balphastar\|^2\right]+\eta_t^2c_5^2-\epsilon_t'(1-\epsilon_t')\EE[\|p_t+q_t\|^2]
    \end{align}
    for some constant $c_5$. For the third term above, we note that  
    \begin{align*}
        -\|p_t+q_t\|^2&= -\|p_t\|^2-\|q_t\|^2-2\langle p_t,q_t\rangle \\
        &\leq -\|p_t\|^2-\|q_t\|^2+2\|p_t\|\|q_t\|\\
        &\leq -\|p_t\|^2+c_6\|q_t\|.
    \end{align*}
    Note that $p_t=g(\balpha_t)-\balpha_t=\gamma(\bA\xstar(\balpha_t)-\bellstar)$. Hence the last inequality follows from the compactness of $\XX$ for some constant $c_6>0$. This implies $
        \EE\left[-\|p_t+q_t\|^2\right]\leq \EE\left[-\|p_t\|^2\right]+2c_7\sqrt{\eta_t},$
    for some constant $c_7>0$. This inequality follows from Lemma \ref{lemma:inner} and since $\EE[\|q_t\|]^2\leq \EE[\|q_t\|^2]$. 
    
    Using our assumption that $\epsilon_t$ goes to zero, we have that $0<\epsilon_t'<1/2$ for all $t>T_0$. For the rest of the proof, we pick $T_0$ to be large enough such that  $\epsilon_t^2\leq C_1\eta_t^3$, $\epsilon_t<\gamma/2$ and $\eta_t<C_2$ for all $t>T_0$, for appropriate $C_1$ and $C_2$. 
    Returning to \eqref{split22}, for all $t>T_0$, we have
    \begin{align*}
        &\EE\left[\|\balpha_{t+1}-\balphastar\|^2\right]\\
        &\leq\EE\left[\|\balpha_t-\balphastar\|^2\right]+\eta_t^2c_5^2+\epsilon_t'\sqrt{\eta_t}c_7-(\epsilon_t'/2)\EE[\|p_t\|^2].
    \end{align*}
    This gives us 
    \begin{align*}
        (\epsilon_t'/2)\EE[\|p_t\|^2]&\leq \EE[\|\balpha_t-\balphastar\|^2]- \EE[\|\balpha_{t+1}-\balphastar\|^2]\\
        &\;\;\;+\eta_t^2c_5^2+\epsilon_t'\sqrt{\eta_t}c_7.
    \end{align*}
    Summing from $i=T_0$ to $t$ gives us
    \begin{align}\label{almostdone}
        \sum_{i=T_0}^t (\epsilon_i'/2)\EE[\|p_i\|^2]&\leq \EE[\|\balpha_{T_0}-\balphastar\|^2]- \EE[\|\balpha_{t+1}-\balphastar\|^2]\nonumber\\
        &\;\;+\sum_{i=T_0}^t (\eta_i^2c_5^2+\epsilon_i'\sqrt{\eta_i}c_7).
    \end{align}
    The left hand term in the above inequality depends on $\|p_i\|^2$ for all $i\in\{T_0,\ldots,t\}$. We next show a bound for $\|p_i\|^2$ in terms of $\|p_t\|^2$ for all $i\in\{T_0,\ldots,t\}$. For this, we first recall that $\balpha_{t+1}=\balpha_t+\epsilon_t'(p_t+q_t).$ Next, 
\begin{align*}
    &\|p_{t+1}\|^2=\|p_t\|^2+\|p_{t+1}-p_t\|^2+2\langle p_{t+1}-p_t,p_t\rangle\\
        &\leq \|p_t\|^2+\|p_{t+1}-p_t\|^2+\frac{2}{\epsilon_t'}\langle p_{t+1}-p_t,\balpha_{t+1}-\balpha_t-\epsilon_t'q_t\rangle
    \end{align*}
    Now,
    \begin{align*}
        &2\langle p_{t+1}-p_t, \balpha_{t+1}-\balpha_t\rangle\\
        &=2\langle g(\balpha_{t+1})-g(\balpha_t)-\balpha_{t+1}+\balpha_t, \balpha_{t+1}-\balpha_t\rangle\\
        &=2\langle g(\balpha_{t+1})-g(\balpha_t), \balpha_{t+1}-\balpha_t\rangle-2\|\balpha_{t+1}-\balpha_t\|^2\\
        &=\|g(\balpha_{t+1})-g(\balpha_t)\|^2+\|\balpha_{t+1}-\balpha_t\|^2-2\|\balpha_{t+1}-\balpha_t\|^2\\
        &\;\;-\|g(\balpha_{t+1})-g(\balpha_t)-\balpha_{t+1}+\balpha_t\|^2\\
        &=\|g(\balpha_{t+1})-g(\balpha_t)\|^2-\|\balpha_{t+1}-\balpha_t\|^2-\|p_{t+1}-p_t\|^2\\
        &\leq -\|p_{t+1}-p_t\|^2,
    \end{align*}
    where the last inequality uses that $g(\cdot)$ is non-expansive. Then, 
    \begin{align*}
        &\|p_{t+1}\|^2\leq \|p_t\|^2-\frac{1-\epsilon_t'}{\epsilon_t'}\|p_{t+1}-p_t\|^2-2\langle p_{t+1}-p_t,q_t \rangle\\
        &=\|p_t\|^2-\frac{1-\epsilon_t'}{\epsilon_t'}\left\|p_{t+1}-p_t+\frac{\epsilon_t'}{1-\epsilon_t'}q_t\right\|^2+\frac{\epsilon_t'}{1-\epsilon_t'}\|q_t\|^2\\
        &\leq \|p_t\|^2+\frac{\epsilon_t'}{1-\epsilon_t'}\|q_t\|^2.
    \end{align*}
    Using the assumption that $0<\epsilon_t'<1/2$ for all $t>T_0$, this gives us 
    $\|p_{t+1}\|^2\leq \|p_t\|^2+2\epsilon_t'\|q_t\|^2.$
    Taking expectation and using Lemma \ref{lemma:inner}, this results in the bound
    \begin{align*}
        \EE[\|p_i\|^2]&\geq \EE[\|p_t\|^2]-\sum_{j=i}^{t-1}2\epsilon_j'\EE\|q_j\|^2\\
        &\geq \EE[\|p_t\|^2]-\sum_{j=i}^{t-1}2c_8\epsilon_j'\eta_j,
    \end{align*}
    for some constant $c_8>0$. 
    Returning to \eqref{almostdone}, we now have
    \begin{align}\label{almostdone2}
        &\EE[\|p_t\|^2]\sum_{i=T_0}^t (\epsilon_i'/2)\leq \EE[\|\balpha_{T_0}-\balphastar\|^2]\nonumber\\
        &\;\;+\sum_{i=T_0}^t (\eta_i^2c_5^2+\epsilon_i'\sqrt{\eta_i}c_7)+2c_8\sum_{i=T_0}^t\epsilon_i'\sum_{j=i}^{t-1}\epsilon_j'\eta_j.
    \end{align}
    Now we use our assumption that $\epsilon_t^2\leq C_1\eta_t^3$ for all $t>T_0$. Then $\epsilon_i'\sqrt{\eta_i}=O(\eta_i^2)$ and hence the sequence $\{\epsilon_i'\sqrt{\eta_i}\}$ is summable. Finally for the double summation, we use the integral test and Assumption \ref{assu-stepsize} to obtain 
    \begin{align*}
    &\sum_{i=T_0}^t\epsilon_i'\sum_{j=i}^{t-1}\epsilon_j'\eta_j\leq \frac{1}{\gamma^2}\sum_{i=T_0}^\infty\frac{1}{(i+T_2)^\epsilon}\sum_{j=i}^\infty \frac{1}{(j+T_3)^{\epsilon+\eta}}\\
    &\leq \sum_{i=T_0}^\infty\frac{1}{(i+T_2)^\epsilon}\frac{c_9}{(i+T_3)^{\epsilon+\eta-1}}\leq \sum_{i=T_0}^\infty\frac{c_{10}}{(i+T_3)^{4\eta-1}}.
    \end{align*}
    Here $c_9$ and $c_{10}$ are constants and $T_3=\max\{T_1,T_2\}$. The final inequality follows from the fact that $\epsilon_t^2\leq C_1\eta_t^3$ and hence $2\epsilon\geq 3\eta$. Since $\eta_t$ is square summable, $\eta>0.5$ and hence $4\eta-1>1$, so the summation is bounded by a constant. Hence all terms in the right-hand side of \eqref{almostdone2} are bounded by a constant. Then for some $c_{11}>0$, we have
    $$\EE[\|p_t\|^2]\sum_{i=T_0}^t (\epsilon_i'/2)\leq c_{11}.$$
    Now $\sum_{i=T_0}^t (\epsilon_i'/2)\geq (t-T_0+1)\epsilon_t'/2$.
    This finally gives us the bound
    $\EE[\|p_t\|^2]\leq C''/(t\epsilon_t),$
    for all $t>T_0$.
\end{proof}

\section{Simulation Results}\label{sec:simulations}

\begin{figure*}[!h]
\centering
\begin{subfigure}{.3\textwidth}
  \centering
  \includegraphics[width=0.99\linewidth]{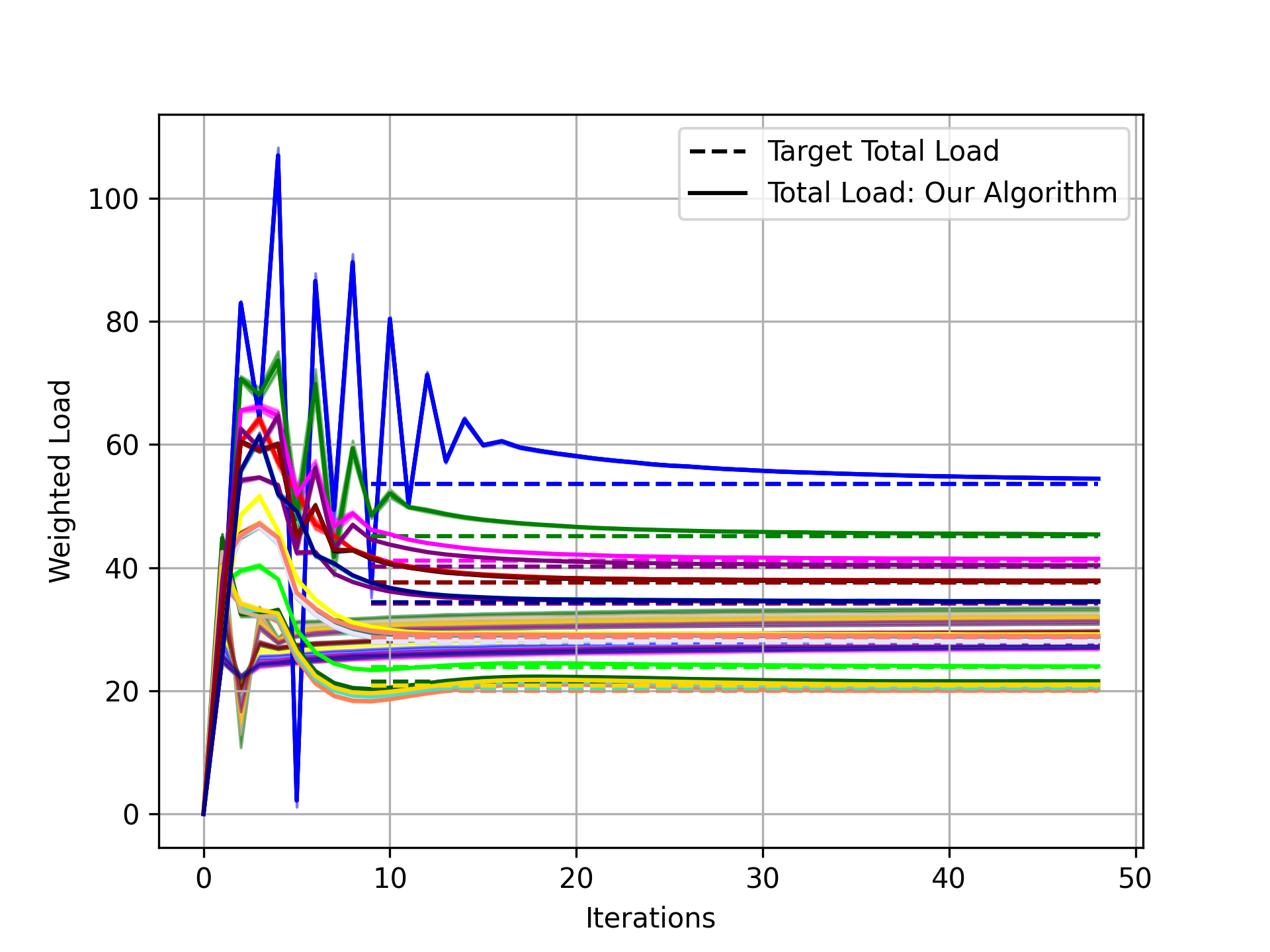}
  \caption{Loads for all resources}
  \label{fig:allo_1}
\end{subfigure}%
\begin{subfigure}{.3\textwidth}
  \centering
  \includegraphics[width=.99\linewidth]{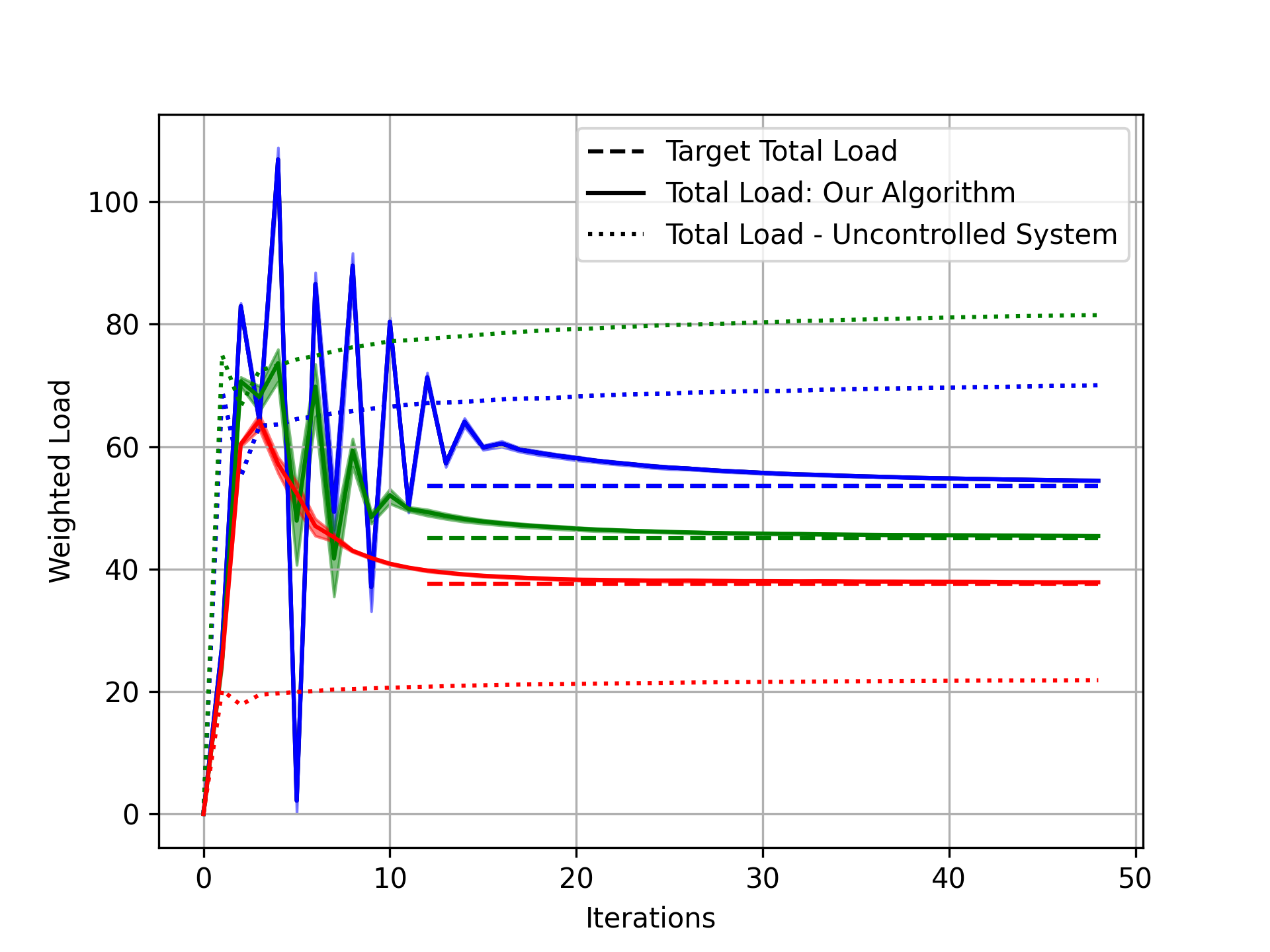}
  \caption{Loads for 3 resources}
  \label{fig:allo_2}
\end{subfigure}
\begin{subfigure}{.3\textwidth}
  \centering
  \includegraphics[width=.99\linewidth]{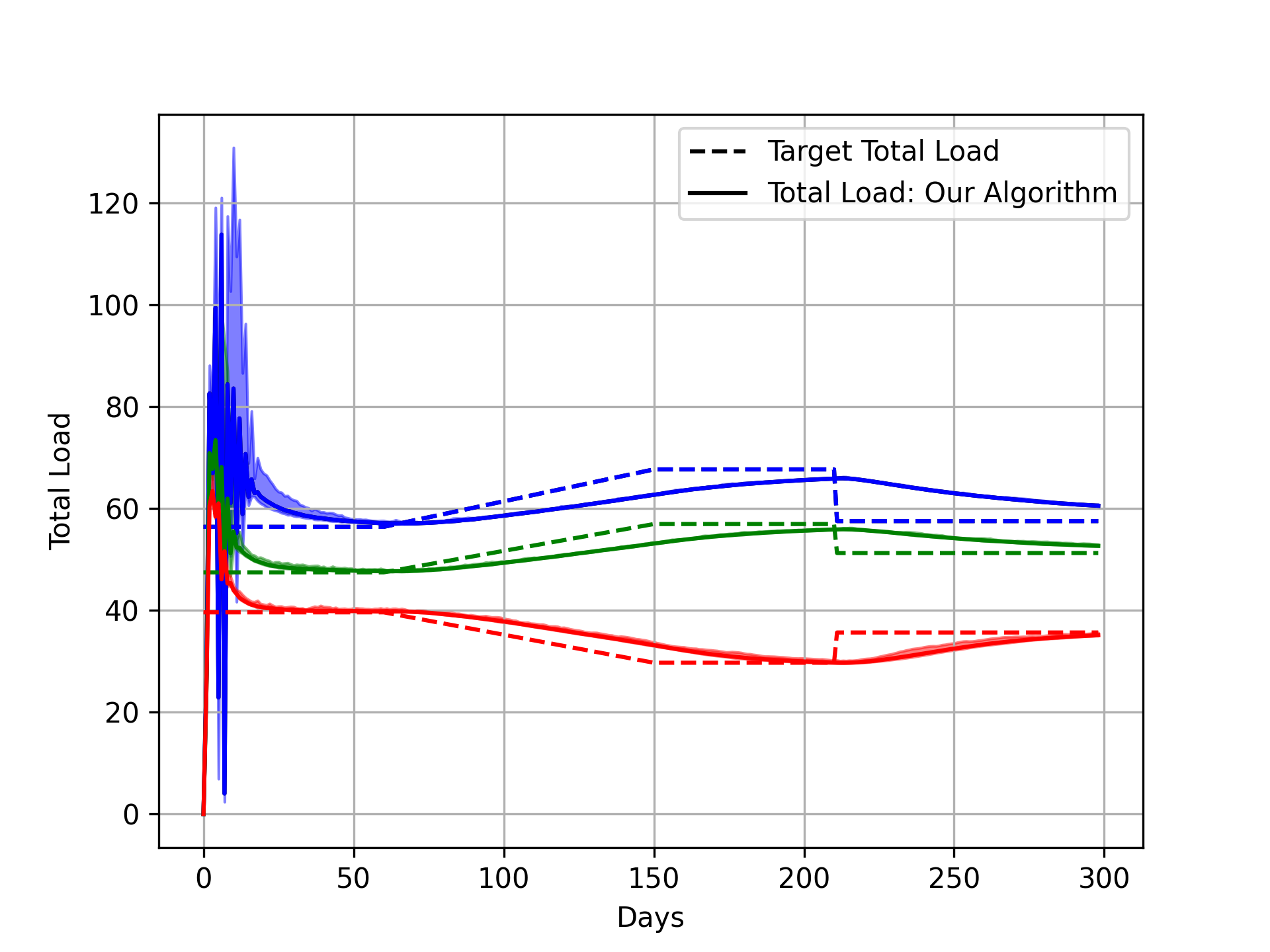}
  \caption{Load balancing for time-varying loads}
  \label{fig:allo_3}
\end{subfigure}
\caption{Load Balancing in Resource Allocation Games}
\label{fig:allo}
\end{figure*} 
\begin{figure*}[!h]
\centering
\begin{subfigure}{.3\textwidth}
  \centering
  \includegraphics[width=0.96\linewidth]{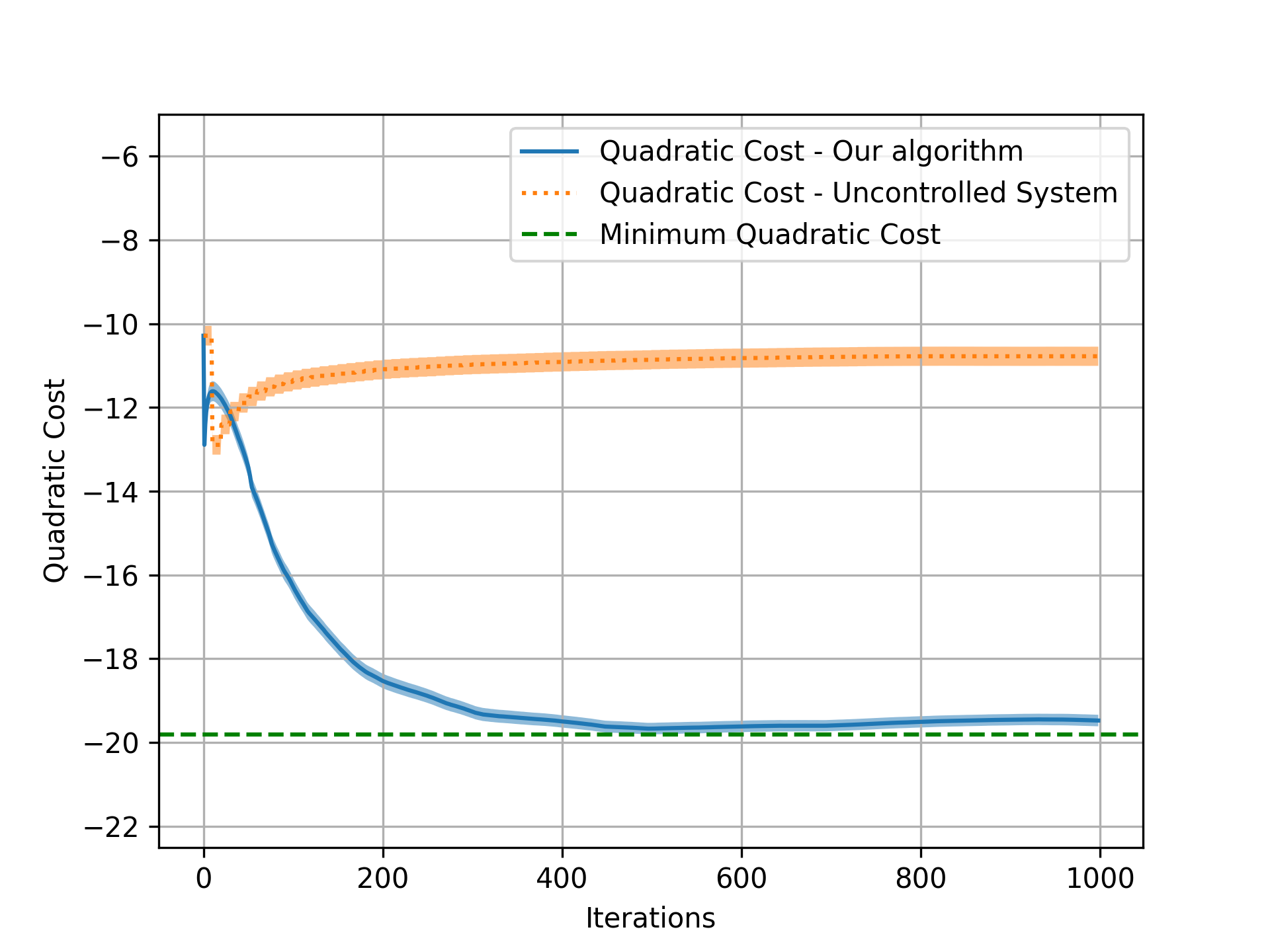}
  \caption{Quadratic cost}
  \label{fig:quad_1}
\end{subfigure}%
\begin{subfigure}{.3\textwidth}
  \centering
  \includegraphics[width=.96\linewidth]{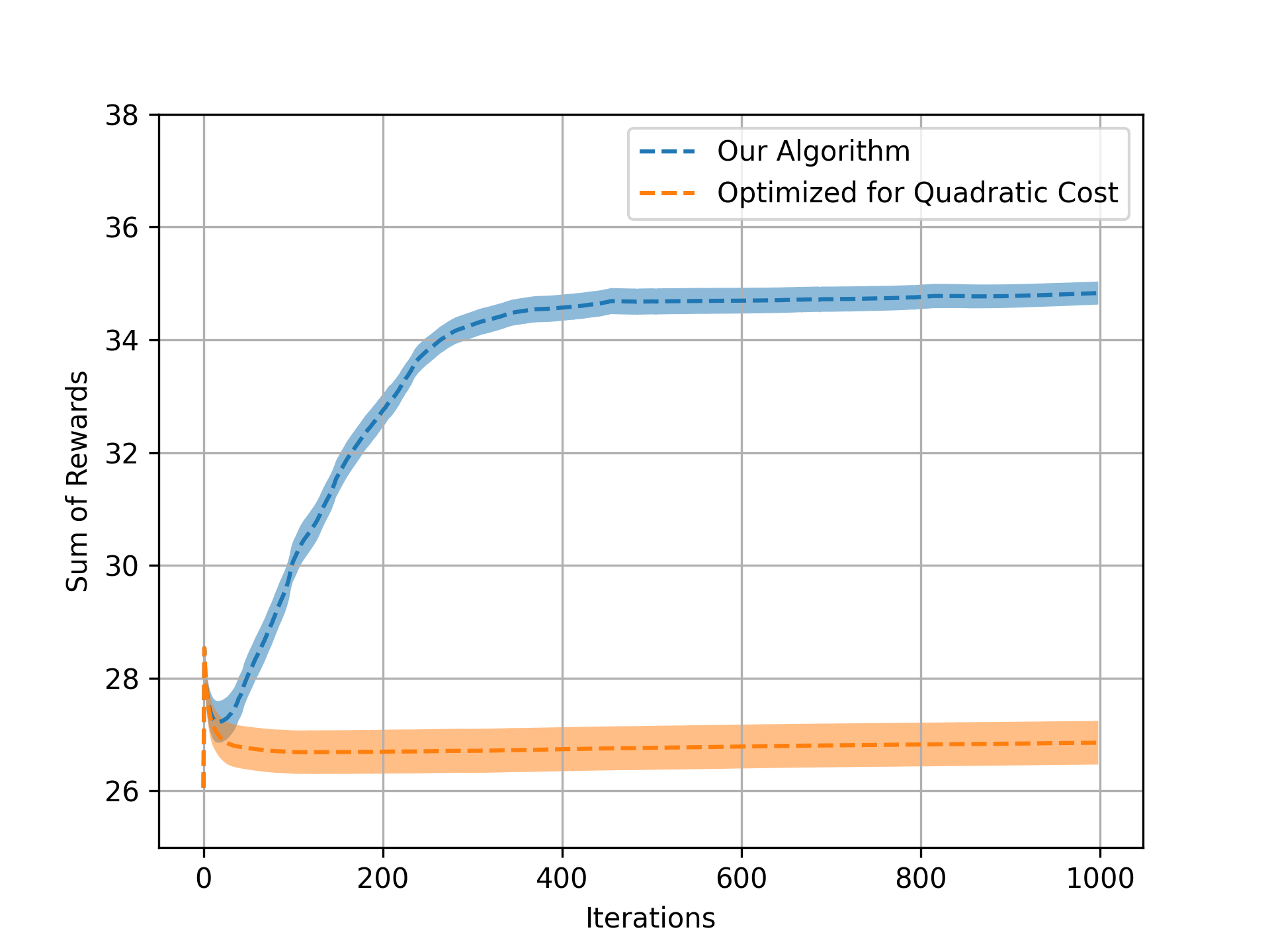}
  \caption{Sum of rewards of all players}
  \label{fig:quad_2}
\end{subfigure}
\begin{subfigure}{.3\textwidth}
  \centering
  \includegraphics[width=.96\linewidth]{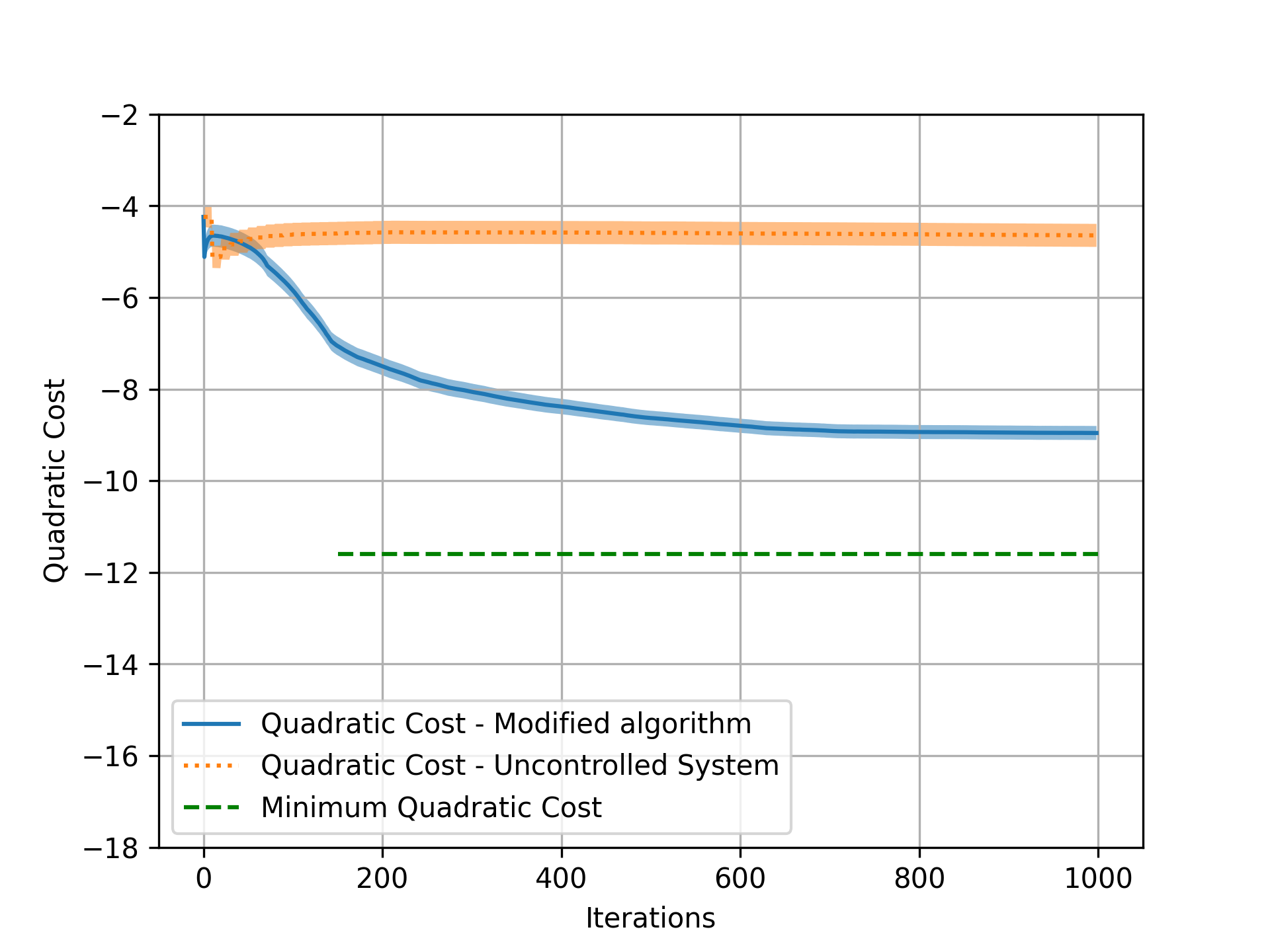}
  \caption{Cost when no $\xx\in\XX$ satisfies $2\boldsymbol{G}\xx+\boldsymbol{\rho}=0$}
  \label{fig:quad_3}
\end{subfigure}
\caption{Quadratic Global Objective}
\label{fig:quad}
\end{figure*} 
In this section, we simulate
the two applications discussed in Section \ref{subsec:Applications}. In all experiments, we use the step-size sequences $\eta_{t}=\frac{1}{\left(t+1\right)^{0.501}}$
and  $\epsilon_{t}=\frac{1}{\left(t+1\right)^{0.753}}$ since they approximately give the best rate (see Theorem  \ref{thm:finite}). To model imperfect gradient information (see Assumption 3), we add Gaussian noise with mean $0$
and variance $\sigma^{2}=\frac{1}{4}$ to the gradients
of each player.
We ran 100 realizations
for each experiment and plotted the average result along with the
standard deviation region, which was always small. In
each realization, $\alpha_{0}^{k}$ and $ x_{n,0}^{k}$
were chosen uniformly and independently at random on $\left[0,2\right]$
and $\left[0,0.1\right]$, respectively. 

We simulated a demand-side management game (Figure \ref{fig:allo}) \cite{chen2014autonomous} with $N=1000$ and $d=24$ hours. The action sets were $\XX_n=\{\xx_n|0\leq x_n^{i}\leq \mu_{n}^i, \sum_{i=1}^dx_{n}^i\leq\mu_n\}$, and $\mu_n^i$ was sampled independently for each $n,i$ from a Gaussian distribution with mean $\hat{\mu}^i$ and variance $0.5$. Here $\hat{\mu}^i$ denote the typical hourly energy consumption limit of a residential customer as given in \cite[Fig.~1]{chen2014autonomous}. 
Similarly, $\mu_n$ was sampled independently for each $n$ from a Gaussian distribution with mean $\hat{\mu}$ and variance $1$, where $\hat{\mu}$ is the typical daily consumption limit. Let $s_i=\sum_{n=1}^N x_n^i$ be the energy consumption at the $i$-th hour. We use the reward function $r_n(\xx)=\sum_{i=1}^d \omega_{n}^ix_{n}^i-0.3(x_n^i)^2-0.01x_n^is_i^2.$
Here $\omega_n^i$ was sampled independently for each $n,i$ from a Gaussian distribution with mean $\hat{\omega}^i$ and variance $0.5$. The value of energy at hour $i$, $\hat{\omega}^i$, was adapted from the demand characteristics of residential consumers \cite{power-book}. The objectives $s_i=\ell_i^*$ were chosen to shave the peak of the uncontrolled system \cite{chen2014autonomous}, with the following target loads:
\begin{align*}
    \bell^*=[&53.58, 45.09, 37.63, 34.32, 41.16, 28.95, 37.55, 34.18,\\ &27.99, 28.15, 28.71, 28.37,
 19.96, 27.80,  23.83, 27.93,\\ &20.45, 21.44, 20.86, 28.67, 20.89, 40.12, 37.64, 34.32].
\end{align*}

In Fig. \ref{fig:allo_1}, we plot the weighted load for each resource along with their target loads. We can observe that the system converges to a NE where the loads are balanced. In Fig. \ref{fig:allo_2}, we plot the weighted load for the first 3 resources. To study the effectiveness of our algorithm, we also plot the loads in the case of an uncontrolled system, i.e., where $\balpha=
\boldsymbol{0}$. It can be observed that if the system is uncontrolled, some resources are under-utilized while others are over-utilized. 

For the quadratic global objective (Figure \ref{fig:quad}), we consider a quadratic game. In a quadratic game, the reward of each player is of the form $r_{n}\left(\boldsymbol{x}\right)=\boldsymbol{x}^{T}\boldsymbol{Q_n}\boldsymbol{x}+\boldsymbol{c_n}^T\boldsymbol{x}_{n}$ for some negative definite matrix $\boldsymbol{Q_n}\in\RR^{Nd\times Nd}$ and $\boldsymbol{c_n}\in\RR^d$. 

We work with $N=100$ players and $d=5$, and action sets $\{\xx_n\mid0\leq\sum_{i=1}^d x_n^i\leq 10\}$ for each player $n$. Recall that the global cost is given by $\xx^\top \boldsymbol{G}\xx+\boldsymbol{\rho}^T\xx$. To generate $\boldsymbol{G}$, we first generate random $\boldsymbol{H_0}\in\RR^{Nd\times (Nd-5)}$ with each element of $\boldsymbol{H_0}$ sampled independently from a standard Gaussian distribution. We use $\boldsymbol{G}=\boldsymbol{H_0}\boldsymbol{H_0}^T$. A random $\boldsymbol{\rho}$ is obtained with each element sampled from a Gaussian distribution with mean 0 and variance $0.5$. The above matrices $\boldsymbol{Q_n}$ are obtained by sampling $\boldsymbol{H_n}\in\RR^{Nd\times Nd}$ and using $\boldsymbol{Q_n}=\boldsymbol{H_n}\boldsymbol{H_n}^T$. The vectors $\boldsymbol{c_n}$ are sampled in a similar way to $\boldsymbol{\rho}$.

Fig. \ref{fig:quad_1} shows that our algorithm converges to the optimal global cost, which is significantly better than the global cost of the uncontrolled system. Even more, Fig. \ref{fig:quad_2} shows an advantage of our scheme over a centralized scheme that optimizes the global cost $\Phi(\xx)$ directly, i.e.,  where each player updates their action using $\nabla \Phi(\xx)$. While both approaches minimize the global cost, our algorithm results in a significantly better sum of rewards $\sum_{n=1}^{N}r_n(\boldsymbol{x})$. The two plots together show that our algorithm can optimize the global cost under the soft constraint of maintaining a high sum of local player rewards.

In Figures \ref{fig:allo_3} and \ref{fig:quad_3}, we study the performance of our algorithm beyond our theoretical model and guarantees. In Fig.~\ref{fig:allo_3}, the target loads $\bellstar_t$ vary over time. The objective now is $\sum_{n=1}^Nx_{n,t}^i=\ell_{i,t}^*$. Our algorithm  easily tracks the time-varying loads under both smoothly varying loads and abruptly changing load profiles. In Fig.~\ref{fig:quad_3}, we study the case where the linear constraints are infeasible, i.e., $\bA\xx=\bellstar$ is not satisfied by any $x\in\XX$. In this case, we make a slight modification to our algorithm by projecting the control input $\balpha$ onto a sphere after each iteration. This is required to ensure the stability of the algorithm. Although it is no longer possible to achieve the optimal cost, our modified algorithm still demonstrates stable behavior and convergence to an action profile with a lower quadratic cost as compared to the uncontrolled system.

\section{Conclusion}\label{sec:conclusion}

We studied how to control a strongly monotone game that is unknown to the manager. The goal of the manager is to ensure the game satisfies target linear constraints at NE. Our simple algorithm
only requires the manager to observe the instantaneous constraint violation.
We prove that our algorithm guarantees convergence to an $\xstar$ which satisfies  $\bA\xstar=\bellstar$ by adjusting control inputs $\balpha$ in real-time. Further, we show that the convergence rate for our algorithm is $O(t^{-1/4})$. Our two time-scale algorithm does not wait for the players to converge to a NE but uses the feedback from every turn. This is an encouraging
result since our algorithm is easy to implement and maintains the
users' privacy. The manager does not observe individual actions, and players do not reveal their reward functions, which
might even be unknown to them. 

Technically, we analyze how the NE of the game
changes when the manager changes $\balpha_{t}$. Our algorithm
essentially controls the unknown game such that its NE becomes desirable
from an optimization point of view.  Since convergence to
NE is relatively fast, this has the potential to accelerate general
cooperative multi-agent optimization \cite{bistritz2020cooperative}.

Controlling unknown games with bandit feedback is a new online learning paradigm. The game to be controlled is a non-stationary
but highly structured bandit. Viewing game control as a bandit problem can result in new, powerful game control techniques. Generalizing our game control setting beyond linear constraints can help tackle a wider class of problems. 

\appendices

\section{Proofs from Section \ref{sec:Algorithm}}
\subsection{Proof for Theorem \ref{thm:converge}}

For a fixed $\balpha$, the ODE limit of iteration \eqref{x-iter} is 
$\dot{\xx}(t)=F(\xx(t), \balpha)+\bb_\XX(\xx(t)).$
Here $\bb_\XX(\xx(t))$ is the boundary projection term required to keep $\xx(t)$ inside $\XX$ at all terms. Since the game is strongly monotone for any $\balpha$, \cite[Theorem 9]{rosen1965existence} shows that for a fixed $\balpha$, there exists a unique globally stable equilibrium of the above ODE, which is the NE $\xstar(\balpha)$. Additionally, we have shown that $\|\balpha_t\|$ has bounded expectation in Lemma \ref{lemma:bounded} for all $t$, which implies that with probability 1, $\|\balpha_t\|<\infty$ for all $t$.  Then Part (a) of Lemma \ref{lemma:xstar-prop} shows that $\xstar(\balpha)$ is Lipschitz continuous in $\balpha$. Then \cite[Lemma 8.1]{VSBbook} shows that $\xx_t\to\xstar(\balpha_t)$. Now, we can write iteration \eqref{alpha-iter} as
$$\balpha_{t+1}=\balpha_t+\epsilon_t'(g(\balpha_t)-\balpha_t)+\epsilon_t'\gamma\bA(\xx_t-\xstar(\balpha_t)),$$
where $\epsilon_t'=\epsilon_t/\gamma$. The ODE limit for this iteration is 
\begin{equation}\label{ode-alpha}
    \dot{\balpha}(t)=g(\balpha(t))-\balpha(t).
\end{equation}
 The set of fixed points of $g(\cdot)$ is $\Nopt$, which is non-empty as shown in Lemma \ref{lemma:existence}. Let $\balphastar$ be any such point. \cite[Theorem 3.1]{VSB-nonexp} shows that the solutions of $\dot{\boldsymbol{y}}(t)=f(\boldsymbol{y}(t))-\boldsymbol{y}(t)$ converge to the set of fixed points of $f(\cdot)$, when $f$ is non-expansive under $\ell_2$ norm and the set of fixed points of $f(\cdot)$ is non-empty. Since $\balpha(t)$ converges to a point in $\Nopt$, this completes our proof that the iterates $\balpha_t$ converge to $\Nopt$.
\subsection{Proof for Theorem \ref{thm:finite}}
Let $T_0$ be the time after which the bounds in Lemma \ref{lemma:inner} and Lemma \ref{lemma:outer} hold. Then for $t>T_0$,
    \begin{align*}
        &\EE[\|\bA\xx_t-\bellstar\|^2]\\
        &\leq 2\EE[\|\bA(\xx_t-\xstar(\balpha_t))\|^2]+2\EE[\|\bA\xstar(\balpha_t)-\bellstar\|^2]\\
        &\leq 2\|\bA\|^2\EE[\|\xx_t-\xstar(\balpha_t)\|^2]+2\EE[\|\bA\xstar(\balpha_t)-\bellstar\|^2]\\
        &\leq C\left(\eta_t+1/(t\epsilon_t)\right),
    \end{align*}
    for $C=\max\{2\|\bA\|^2C',2C''\}$.

\section*{References}
\begingroup
\setlength{\parskip}{0pt}
\vspace*{-19pt}

\end{document}